\documentclass[12pt,reqno]{article}

\usepackage{graphicx} 
\usepackage{hyperref}
\usepackage{nameref}
\hypersetup{breaklinks=true,
  colorlinks   = true, 
  urlcolor     = blue, 
  linkcolor    = {blue!60!black}, 
  citecolor   = {blue!40!black} 
}
\usepackage{natbib}   
\usepackage{appendix}
\usepackage{setspace}      
\usepackage{enumerate}
\usepackage{mathtools,cases}                  
\usepackage{amsmath,amssymb,amstext} 
\usepackage{amsfonts}           
\usepackage{bm}
\usepackage{float} 
\usepackage{mathrsfs}   
\usepackage{enumitem}
\usepackage{amsmath,amsfonts,amsthm,amssymb}
\usepackage[letterpaper,margin=1in]{geometry}
\usepackage{tikz}
\usetikzlibrary{calc} 
\usetikzlibrary{arrows.meta}
\usetikzlibrary{arrows}
\usetikzlibrary{chains}
\usepackage{pgfplots}
\pgfplotsset{compat=1.18} 
\usepackage{multirow,booktabs}
\usepackage{multirow,makecell}
\usepackage{xtab}
\usepackage{authblk}
\usetikzlibrary{trees}

\setlist{nolistsep,leftmargin=*}

\newtheorem{proposition}{Proposition}
\newtheorem{theorem}{Theorem}

\newtheorem{lemma}{Lemma}
\newtheorem{corollary}{Corollary}

\theoremstyle{definition}
\newtheorem{axiom}{Axiom}
\newtheorem{example}{Example}
\newtheorem{definition}{Definition}
\newtheorem{notation}{Notation}

\begin{document}

\title{\Huge Complexity Aversion\thanks{I am deeply grateful to Hans Haller, Eric Bahel, Matt Kovach, and Sudipta Sarangi for their guidance and encouragement throughout this project. I would also like to thank Yaojun Liu for his suggestions. All errors in this paper are solely the responsibility of the author.}}
\author{ Yuan Gu \thanks{Fuyao Institute for Advanced Study, Fuzhou, China. (E-mail: gyuan@fyust.org.cn).} \quad \quad \quad   Chao Hung Chan \thanks{Department of Economics, Virginia Tech, Blacksburg, USA (e-mail: chanhungc@vt.edu).}}
\date{}

\maketitle 
\begin{abstract}
\noindent ABSTRACT. This paper proposes a model of decision-making under uncertainty in which an agent is constrained in her cognitive ability to consider complex acts. We identify the complexity of an act according to the corresponding partition of state space. The agent ranks acts according to the expected utility net of complexity cost. A key feature of this model is that the agent is able to update her complexity cost function after the arrival of new information. The main result characterizes axiomatically an updating rule for complexity cost function, the Minimal Complexity Aversion representation. According to this rule, the agent measures the complexity cost of an act conditional on the new information by using the cost of another act that gives exactly the same partition of the event but with the lowest ex-ante cost.\\
\textbf{Keywords}: Complexity Aversion, Minimal Complexity Updating
\end{abstract}

\newpage
\section{Introduction} \label{sec1}

\noindent
\subsection{Motivation} \label{sec1.1}

An underlying assumption of classical expected utility theory, that is not explicitly formed as an axiom in the literature (e.g., \cite{savage1972foundations}, \cite{anscombe1963definition}), is that there is no contemplation cost when an agent considers an act. However, in reality, there is often a trade-off between choosing a more complex act that leads to more distinct consequences and saving on the cognitive resources  (or physical resources) of doing so. To account for this complexity aversion, I propose an axiomatic model that incorporates the  complexity costs of acts into the decision-making process.

In this paper, I model a decision maker making a choice under uncertainty who is constrained in her cognitive ability to consider complex acts. I identify $(\romannumeral1)$ to what extent the agent is constrained; $(\romannumeral2)$ how the agent measures the cost of acts after the reduction of uncertainty (due to the fact that after the arrival of new information, she can discard all irrelevant states). 

To illustrates how the complexity costs of acts can play a critical role in decision-making, especially when the consequences of different actions are uncertain or difficult to predict, let's consider a more detailed example of individual investment. Consider an investor who wants to decide her investment plan. A state of the world indicates the state of the economy. Each investment plan is an act, which attaches a payoff to each state. However, there are substantive economic indicators for both microeconomics and macroeconomics. It is very difficult for individual investors to figure out how the economy will perform. As a result, when the investor considers an investment plan, she may be limited in her ability to consider every possible outcome associated with each state of the world that the act gives. Even if a sophisticated investment plan offers different outcomes for each state, the investor may be unwilling to choose it because of the high complexity cost of this plan. Consequently, she may choose a simpler investment plan, such as putting money in the bank. 

How does an agent like the investor I describe above choose acts? Naturally, when the agent chooses an act, she  identifies the complexity of an act $f$ by $\sigma(S^{f})$ which is the $\sigma$-algebra generated by act $f$, where $S^{f}$ is the partition induced by $f$. $f$ is more complex than $g$ if $\sigma(f)$ is finer than $\sigma(g)$. In words, the agent considers $f$ is more complex than $g$ if she has to consider more distinct outcomes that $f$ gives than $g$ does. Thus, if $f$ is too complex compared with $g$, she may choose $g$ even when $f$ yields higher expected utility than $g$.

\begin{example}
            \hypertarget{exp1}{To} fix the idea, suppose uncertainty is modeled by the state space $\Omega=\{\omega_{1}, \omega_{2}, \omega_{3}\}$, where $\omega$ indicates the state of the economy.  There are two investment plans available which can be represented as vectors of state-contingent payoffs as follows: 
            \end{example} 
	        \begin{table}[ht]
            \centering
            \renewcommand{\arraystretch}{2.0}{
	        \setlength{\tabcolsep}{4mm}{
		        \begin{tabular}{|c|c|c|c|l|}
			    \hline
			    $\Omega$&  $\omega_{1}$ & $\omega_{2}$  &  $\omega_{3}$ & $S$\\
			    \hline
			    $f$ & $1$ & $3$ & $4$ & $S^{f}=\big\{\{\omega_{1}\},\{\omega_{2}\},\{\omega_{3}\}\big\}$\\
			    \hline
			    $g$ & $2$ & $2$ & $2$ & $S^{h}=\{\Omega\}$\\
			    \hline
	            \end{tabular}}}
            \caption{Investment Example.}
            \label{tab:table2.1}
            \end{table}

 \noindent Her prior is $\mu=(\frac{1}{3},\frac{1}{3},\frac{1}{3})$. It is easy to see that $f$ is more complex than $g$ and $f$ gives higher expected utility than $g$. However, suppose we observe that the investor chooses $g$, which is not the choice predicted by SEU. We can infer that the investor chooses $g$ because she cannot bear the complexity cost of $f$.
 
 I propose an axiomatic model of an agent who is constrained in her cognitive resources to consider complex acts, like the investor above. The primitive is a collection of preferences relations on acts. One of the key aspects of the model is that the agent chooses an act that maximizes expected utility net of complexity cost.  
 
 I introduce a new axiom, called \hyperlink{ac}{Aversion to Complexity}, that depicts the agent's aversion to complex acts. Suppose for any acts $f,g$ and constant outcome $x$, we observe the following preference
            \begin{equation*}
 	        \alpha x + (1-\alpha)f\sim \beta x + (1-\beta) f,
            \end{equation*}
            for any $\alpha,\beta\in(0,1)$ with $\alpha>\beta$. That is the agent is indifferent to acts $x$ and $f$ if the agent does not have to consider the complexity cost of $f$. Since $x$ is the simplest act, if the agent is indeed complexity averse, we anticipate that she will prefer $\lambda x + (1-\lambda)g$ to the randomization over $f$ and $g$ if $\lambda f + (1-\lambda) g$ is more complex than $g$, i.e., her preference should be
            \begin{equation*}
 	        \lambda x + (1-\lambda)g \succsim \lambda f + (1-\lambda)g
            \end{equation*}
            for any $\lambda\in (0,1)$. In other words, since $\lambda f + (1-\lambda)g$ causes at least as high complexity cost as $\lambda x + (1-\lambda)g$, the agent has to devote more cognitive efforts to process it. However, an SEU agent is indifferent between $\lambda x + (1-\lambda)g $ and  $\lambda f + (1-\lambda)g$. 
 
The above axiom, together with other modified classical SEU axioms characterizes the representation 
            \begin{equation*}
 	        V(f) = \int_{\omega \in \Omega} u(f(\omega))\mu(d\omega)-\mathcal{C}\big(\sigma(f)\big)
            \end{equation*}
 
\noindent where $\mathcal{C}$ is the function used to measure the complexity cost of any act $f$,\footnote{See \hyperref[sec3.2]{Section 3.2} for more discuss of properties of the complexity cost function.} $u$ is a utility index and $\mu$ is a probability measure. The representation suggests that the agent ranks acts according to expected utility net of complexity cost. Since the agent is self-aware of her cognitive constraint, she compares acts by considering the complexity cost which can be identified by the $\sigma$-algebra generated by acts. The formal definition of cost function is given in \hyperlink{def1}{Definition 1}. One key feature of this representation is that I do not assume any specific forms of cost function except for monotonicity.
 
 Furthermore,  suppose that the investor will receive a report about the state of the economy which indicates that the true state must lie in $E = \{\omega_{2}\}$.\footnote{Note that we assume the analyst (or we call observer, consulting manger) and the agent (the investor) have the same and correct understanding of the report.} Now, after the arrival of this new information, we observe that $f$ is chosen. Therefore, we observe preference reversal which cannot be explained by SEU. It is obvious that $S^{f|E} = S^{g|E} = \{\omega_{2}\}$, where $S^{f|E}$ is the partition on $E$ induced by $f$. Thus, we can infer that the investor uses a new cost function to measure the complexity cost of acts after receiving the new information. That is, with the reduction of uncertainty, she is able to process some complex acts that she could not deal with before.  \vskip 0.1in
 
 To model the investor's ex post choice, I propose a possible updating procedure to deal with incoming information. The agent has to decide how to measure the complexity cost of acts after receiving the new information. In my setting, the agent precisely knows that all states outside of an event $E$ can be discarded after the occurrence of $E$.\footnote{This is stated as an axiom in the standard expected utility theory, which is well known as \hyperlink{conse}{Consequentialism}.} Then, it is intuitive to reason that the agent does not have to consider the complexity of an act $f$ outside of $E$. From this point of view, given $E\in\Sigma^{\prime}$ and $\mu\in\Delta(\Omega)$, we suggest the following conditional complexity cost function: 
            \begin{equation*}
 	        \mathcal{C}_{E,\mu}\big(\sigma(f)\big)=min\big\{\mathcal{C}\big(\sigma(h)\big)/ \mu(E): h\in \mathcal{F} , \text{ and }\sigma(h|E)=\sigma(f|E)\big\},
            \end{equation*}
        where $\mathcal{F}$ is the set of all acts and $\sigma(f|E)$ is the $\sigma$-algebra generated by act $f\in \mathcal{F}$ on $E$. $\mathcal{C}$ is the complexity cost function that is defined in \hyperlink{def1}{Definition 1}. If an agent's conditional complexity cost function is formed as above, I say that she is a minimal complexity updater. The minimal complexity updating rule indicates that the agent measures the complexity cost of an act $f$ conditional on $E$ by using the cost of act $h$ that gives exactly the same partition as $f$ on $E$, but has lowest cost on $\Omega$.   
 
 Then, the conditional preference $\succsim_{E}$ is represented by
            \begin{equation*} 
 	        V(f|E) = \int_{\omega \in \Omega} u(f(\omega))\mu_{E}(d\omega)-\mathcal{C}_{E,\mu}\big(\sigma(f)\big),
            \end{equation*}  
 
\noindent where  $\mu_{E}$ is the Bayesian update of $\mu$ conditional on $E$ and $\mathcal{C}_{E,\mu}$  is the conditional cost function we discussed above. The agent constructs her conditional preferences  by two updating rules. First, she updates her prior beliefs by Bayes' rule. Second, she updates her measure of complexity costs of acts by minimal complexity updating rule. 
 
To identify the behavior of a complexity averse agent after the arrival of new information, I introduce the following key axiom, called \hyperlink{mcu}{Minimal Complexity Updating}. Suppose for any non-null event $E$, $f\in\mathcal{F}$ and $x,z\in\bar{\mathcal{F}}$, we observe the following preference
            \begin{equation*} 
 	        fEz\succsim_{E}x,
            \end{equation*}  
            That is the agent prefers $fEz$ to $x$ conditional on $E$. Then given another act $fEz^{\prime}$, where $\sigma(fEz)\subset \sigma(fEz^{\prime})$. If the agent is a minimal complexity updater, we anticipate that her preference between $fEz^{\prime}$ and $x$ should be
              \begin{equation*} 
 	        fEz^{\prime}\succsim_{E}x.
            \end{equation*} 
The two acts $fEz$ and $fEz^{\prime}$ are the same on $E$, but $fEz^{\prime}$ is more complex than $fEz$ on $\Omega$. If the agent uses a minimal complexity cost function to measure the complexity cost of an act after the arrival of new information, the above two acts should be assigned the same complexity cost.

It is well known that \hyperlink{conse}{Consequentialism} and \hyperlink{dc}{Dynamic Consistency} imply Bayes' rule. However, the minimal complexity updating rule suggests that the agent's behavior will depart from \hyperlink{dc}{Dynamic Consistency}. Thus, instead of \hyperlink{dc}{Dynamic Consistency}, we suggest a novel axiom, called \hyperlink{dca}{Dynamic Complexity Aversion}. 
 
Given any act $g$, we can construct $h=gEx$ such that $\mathcal{C}(h)\leq \mathcal{C}(gEx^{\prime})$ for any constant outcome $x^{\prime}$. Suppose we observe the agent's preference to be
            \begin{equation*}
 	        fEx\succsim gEx,
            \end{equation*}
where $gEx$ is the act $g$ outside of $E$ but yields $x$ for all $\omega\in E$. \hyperlink{dca}{Dynamic Complexity Aversion} requires that she would not prefer $g$ to $f$ if $E$ actually happens, i.e. her conditional preference should be
            \begin{equation*}
 	        f\succsim_{E} g.
            \end{equation*}
Therefore,  \hyperlink{dca}{Dynamic Complexity Aversion} allows for violations of \hyperlink{dc}{Dynamic Consistency}.
 
A noteworthy remark is that, I can impose different structures on the conditional complexity cost function. By doing this, appropriate behavioral properties (axioms) should be formalized. 

In the above example, I discuss how my model can be used to describe those kind of contexts in which agents are aware of better choices but constrained in cognitive ability. This model does not require the observer to know an agent's complexity cost function --- which can be identified by the act-choice data. Moreover, the example shows the importance of new information that enables an agent remeasures the complexity of an act by means of a conditional complexity cost function. 

I extend the existing studies in two ways. First, I propose a framework of choice under uncertainty that models a decision maker who faces constraints on her cognitive ability of choosing an act. The agent is able to identify the complexity of an act according to the corresponding partition of state space. In the literature, a partition of the state space is described as an information structure\footnote{It is widely discussed in the rational inattention literature. This paper also motivated by studies of rational inattention.} (e.g., \cite{ellis2018foundations}), or an agent's understanding of uncertainty (e.g.,\cite{ahn2010framing}, \cite{minardi2019subjective}). Here, I call it complexity. The complexity costs can be identified from act-choice data. Second, I propose a novel conditional complexity cost function and provide axiomatic foundations for this updating rule of the complexity cost function. In this model, preference reversal is not explained by belief updating but instead by the agent's updated complexity cost function. When the agent is told that the true state must belong to a particular event, she needs not consider the complexity of an act outside of this event. Thus, if the complexity of two acts differs only outside of an event, then the complexity costs are the same for the two acts after the arrival of new information. If the agent discards all irrelevant states after receiving the new information, but still uses the same level of cognitive efforts to process acts, her choice under this circumstance can be viewed as arising from cost-measuring biases.

\indent I also show that how rational inattention can be embedded into our framework. The main challenge here is that, in the rational inattention setting, the agent cannot observe the new information clearly, she has to decide which state is possible. Moreover, I show that with the conditional cognition cost function, an inattentive agent can reallocate her attention after the occurrence of an event. Also, I develop an application to a problem of optimal contract design. Moreover, I suggests a new competitive equilibrium concept in which consumers are complexity averse.\footnote{See \hyperref[appD]{Appendix D} for more details.}  

\subsection{Experimental and Empirical Evidence} \label{sec1.2}

In this subsection, I provide evidence for complexity aversion. \hyperref[sec2]{Section 2} introduces the formal setup. \hyperref[sec3]{Section 3} presents the general model: the axioms that define preference relations and the representation theorem. \hyperref[sec4]{Section 4} includes the discussion of the comparison between different agents. \hyperref[sec5]{Section 5} develops two applications to optimal contract design and attention reallocation. \hyperref[sec6]{Section 6} discusses related theoretical literature.

Indeed, there are lots of experimental evidence suggesting that agents are complexity averse. Early studies, focusing on testing SEU theories, show that if lotteries have the same number of outcomes, then the experimental data can be well explained by SEU. If lotteries have different numbers of outcomes, most of the choices violate SEU (e.g., \cite{conlisk1989three}, \cite{harless1991utility}, and \cite{sopher1993test}). \cite{moffatt2015heterogeneity} document the evidence that agents do exhibit complexity aversion, but the degree of complexity aversion is decreasing with more rounds of experiment, where complexity is defined in terms of the number of different outcomes in the lottery.\footnote{There is also other definition of complexity. For example, \cite{mador2000complexity} and \cite{puri2018preference} define complexity as the size of a lottery's support.} In this paper, I introduce an updating rule of complexity cost function which can be used to explain the reduction of complexity aversion. Other similar studies also use the above definition of complexity, such as \cite{huck1999risk}, \cite{sonsino2002complexity}, and \cite{sonsino2001preference}.  Instead of lotteries, our model studies preferences over acts.

Regarding evidence out of the lab. \cite{garrod2008competition} find evidence that customers'  loyalty to one brand may be induced by complexity aversion to various products. \cite{anagol2017understanding} find that life insurance agents tend to suggest customers buy a dominated product such as whole life insurance instead of a combination of investments. \cite{egan2019brokers} finds similar evidence in the bond market. They show that brokers advise consumers to purchase dominated bonds. 

\section{Setup} \label{sec2}

\indent There is a (nonempty) finite set $\Omega$ of states of the world, and elements $E,E^{\prime}$ in $\Sigma=2^{\Omega}\backslash\{\emptyset\}$, referred to as \textsl{events}. I denote by $\Delta(\Omega)$ the set of all probability measures on $\Omega$. The $\mu \in \Delta(\Omega)$ are called \textsl{beliefs}. Let $X$ denote the set of \textsl{consequences}, which is assumed to be a convex subset of a vector space (see \cite{maccheroni2006ambiguity}). For example, this is the case if $X$ is the set of all lotteries on a set of prizes (this is the classical setting of \cite{maccheroni2006ambiguity}). 

\indent The following notations are about information partition. Let $\mathbb{P}$ denote the set of partitions of $\Omega$, denoting $S\in \mathbb{P}$ for a partition of $\Omega$, i.e. $S=\{s_{1},s_{2},\cdots,s_{L}\}$ with $\cup s_{l}= \Omega$, $s_{l}\neq \emptyset$ for all $l\in \{1,\cdots,L\}$, and $s_{i}\cap s_{j}=\emptyset$ for all $i\neq j \in \{1,\cdots,L\}$. For any partitions $S,S^{\prime}\in \mathbb{P}$, I denote by $\sigma(S)$ the algebra generated by partition $S\in\mathbb{P}$, and I say partition $S$ is finer than $S^{\prime}$ if $\sigma(S^{\prime})\subset \sigma(S)$. 

\indent Let $\mathcal{F}$ denote the set of functions $f: \Omega \rightarrow X$, which are referred to as \textsl{acts}. That is, an act is a function attaching a consequence to each state of the world, e.g., an amount of money. Let $\mathcal{F}^{S}$ denote the set of acts that respect the partition $S$, i.e, $\sigma(f)=\sigma\big(\{f^{-1}(x): x\in f(\Omega)\}\big)=\sigma(S^{f})$ where $S^{f}$ is the partition corresponding to act $f$.\footnote{With $\sigma$-algebra, the model can be extended to a more generalized setting where the state space is infinite.} That is, for any $f\in \mathcal{F}^{S}$, and any $\omega, \omega^{\prime}\in E\in S^{f}$, we have $f(\omega)=f(\omega^{\prime})$. So we obtain $\mathcal{F}=\bigcup_{S\in\mathbb{P}}\mathcal{F}^{S}$. Furthermore, for any $f,g\in\mathcal{F}$, if $\sigma(g)\subset \sigma(f)$, I say $f$ is more complex than $g$. Similarly, for any $f\in \mathcal{F}$, and any $E\in \Sigma$, I use $\sigma(f|E)$ to denote the $\sigma$-algebra that is generated by act $f$ conditional on the event $E$, and denote by $S^{f|E}$ the corresponding partition of act $f$ conditional on $E$. Following a standard abuse of notation, I denote by $x\in \mathcal{F}$ the \textsl{constant act} yielding $x\in X$ in every state. $\overline{\mathcal{F}}$ denotes the set of all constant acts. For any two acts $f, g \in \mathcal{F}$, and for any event $E \in \Sigma$, let $fEg$ denote the act that returns $f(\omega)$ for every $\omega \in E$, and returns $g(\omega)$ for every $\omega \in \Omega \backslash E$. The linear structure of $X$ allows mixtures to be defined as following: for any $f,g\in\mathcal{F}$, and $\alpha\in[0,1]$, a state-wise mixture of two acts $f,g \in \mathcal{F}$ is $\alpha f + (1-\alpha) g$  which is identified as $(\alpha f + (1-\alpha) g)(\omega):=\alpha f(\omega) + (1-\alpha) g(\omega)$. Additionally, for any $f\in\mathcal{F}$ and $E\in \Sigma$, let $f(E):=\{x: f(\omega)=x \text{ and } \omega\in E\}$. 

\indent  The primitive is a class of preference relations $\{\succsim_{E}\}_{E\in\Sigma}$ on $\mathcal{F}$. The agent observes the realization of event $E\in\Sigma$, then chooses an act that maximizes the conditional expected utility. Correspondingly, $\succsim_\Omega$ (the ex ante preference relation) is the case when the agent receives no information, or before he receives it. 

\section{Complexity Aversion} \label{sec3}

\subsection{Foundations} \label{sec3.1}

            I start by presenting axioms for the preference order defined on the set $\mathcal{F}$. Conditional preference $\succsim_{E}$ has the same characteristics. 

            \begin{axiom}[\hypertarget{wo}{Weak order}]
            $\succsim$ \textsl{is reflexive, transitive and complete}. 
            \end{axiom}

            In order to represent the agent's preferences with a utility function, it is essential that her preferences exhibit transitivity and completeness. These properties do not necessarily require the agent to have a full understanding of acts as functions that attach outcomes to states. The Weak Order axiom simply requires that the agent will make a certain choice and will not randomly make such decisions.

            \begin{axiom}[\hypertarget{continuity}{Continuity}] 
            \textsl{For any} $f, g, h\in \mathcal{F}$, $\alpha\in[0,1]$, \textsl{and} $\sigma(\alpha f+(1-\alpha) g)=\sigma(f)$, \textsl{the following sets are closed}: 
            \begin{equation*}
	        \{\alpha \in[0,1]: \alpha f+(1-\alpha) g \succsim h\} \quad \textsl{ and } \quad\{\alpha \in[0,1]: h \succsim \alpha f+(1-\alpha) g\}.
            \end{equation*} 
            \end{axiom}

            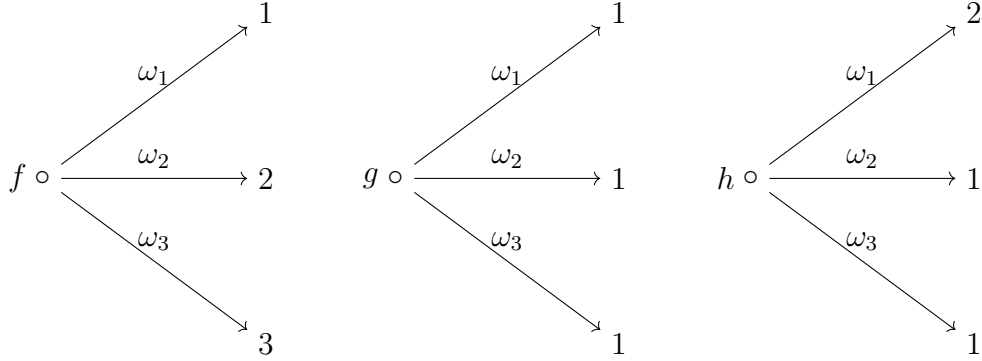
\begin{figure}[ht]
            \centering
            \begin{minipage}{0.3\textwidth} 
            \begin{tikzpicture}[scale=1.1]

            \node (a) at (0,0) {$f$};    
            \node (b) at (0.3,0) {$\circ$};
            \node (z1) at (3,2) {$1$};    
            \node (z2) at (3,0) {$2$};    
            \node (z3) at (3,-2) {$3$};

            \draw [->] (b) -- (z1) node[midway,above] {$\omega_{1}$};     
            \draw [->] (b) -- (z2) node[midway,above] {$\omega_{2}$};     
            \draw [->] (b) -- (z3) node[midway,above] {$\omega_{3}$}; 

            \end{tikzpicture}
            \end{minipage}
            \begin{minipage}{0.3\textwidth} 
            \begin{tikzpicture}[scale=1.1]

            \node (a) at (0,0) {$g$};    
            \node (b) at (0.3,0) {$\circ$};
            \node (z1) at (3,2) {$1$};    
            \node (z2) at (3,0) {$1$};    
            \node (z3) at (3,-2) {$1$};

            \draw [->] (b) -- (z1) node[midway,above] {$\omega_{1}$};     
            \draw [->] (b) -- (z2) node[midway,above] {$\omega_{2}$};     
            \draw [->] (b) -- (z3) node[midway,above] {$\omega_{3}$}; 

            \end{tikzpicture}
            \end{minipage}
            \begin{minipage}{0.3\textwidth} 
            \begin{tikzpicture}[scale=1.1]

            \node (a) at (0,0) {$h$};    
            \node (b) at (0.3,0) {$\circ$};
            \node (z1) at (3,2) {$2$};    
            \node (z2) at (3,0) {$1$};    
            \node (z3) at (3,-2) {$1$};

            \draw [->] (b) -- (z1) node[midway,above] {$\omega_{1}$};    
            \draw [->] (b) -- (z2) node[midway,above] {$\omega_{2}$};     
            \draw [->] (b) -- (z3) node[midway,above] {$\omega_{3}$}; 

            \end{tikzpicture}
            \end{minipage}
            \caption{Three acts $f$, $g$, and $h$.}
            \label{fig:figure1}
            \end{figure}

            Full continuity is too strong in my setting. For example, consider three acts in \hyperref[fig:figure1]{Figure 1}. Let's construct two acts $f_{n} = \frac{1}{n}h + (1-\frac{1}{n})f$ and $g_{n} = \frac{1}{n}h + (1-\frac{1}{n})g$ such that $f_{n}\succsim g_{n}$ for all $n\in\mathbb{N}$.\footnote{With three states, the complexity cost is $\mathcal{C}=\{\mathcal{C}(\{\Omega\}),\mathcal{C}(\{\{\omega_{1},\omega_{2}\},\{\omega_{3}\}\}),\mathcal{C}(\{\{\omega_{1}\},\{\omega_{2}\},\{\omega_{3}\}\})\}$. Then the costs of acts are defined trivially.} It is easy to see that $f_{n} \rightarrow f$ and $g_{n} \rightarrow g$ as $n\rightarrow \infty$. If we don't restrict the complexity level of $f_{n}$ and $g_{n}$, it is possible that the agent prefers $g$ to $f$ when $\mathcal{C}(f)\gg \mathcal{C}(g)$. The reason behind this is that the complexity level of the mixture of two acts is ambiguous. In this case, her choice violates continuity.
\ \\
            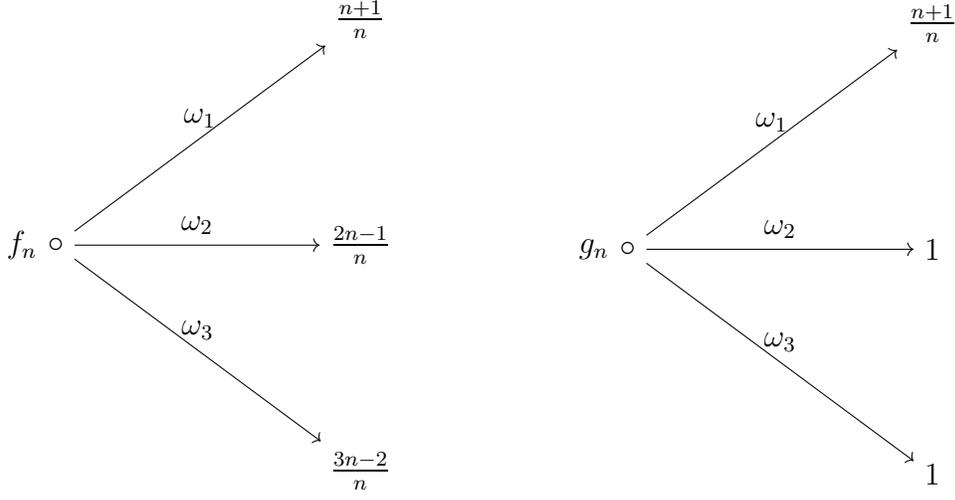
\begin{figure}[ht]
            \centering
            \begin{minipage}{0.49\textwidth} 
            \begin{tikzpicture}[scale=1.5]

            \node (a) at (0,0) {$f_{n}$};  
            \node (b) at (0.3,0) {$\circ$};
            \node (z1) at (3,2) {$\frac{n+1}{n}$};    
            \node (z2) at (3,0) {$\frac{2n-1}{n}$};    
            \node (z3) at (3,-2) {$\frac{3n-2}{n}$};

            \draw [->] (b) -- (z1) node[midway,above] {$\omega_{1}$};     
            \draw [->] (b) -- (z2) node[midway,above] {$\omega_{2}$};     
            \draw [->] (b) -- (z3) node[midway,above] {$\omega_{3}$}; 

            \end{tikzpicture}
            \end{minipage}
            \begin{minipage}{0.49\textwidth} 
            \begin{tikzpicture}[scale=1.5]

            \node (a) at (0,0) {$g_{n}$};    
            \node (b) at (0.3,0) {$\circ$};
            \node (z1) at (3,2) {$\frac{n+1}{n}$};    
            \node (z2) at (3,0) {$1$};    
            \node (z3) at (3,-2) {$1$};

            \draw [->] (b) -- (z1) node[midway,above] {$\omega_{1}$};     
            \draw [->] (b) -- (z2) node[midway,above] {$\omega_{2}$};     
            \draw [->] (b) -- (z3) node[midway,above] {$\omega_{3}$}; 

            \end{tikzpicture}
            \end{minipage}
            \caption{Two acts $f_{n}$ and $g_{n}$.}
            \label{fig:figure2}
            \end{figure}

            \begin{axiom}[\hypertarget{wci}{Weak Certainty Independence}] 
            \textsl{For any} $f,g\in \mathcal{F}$, $x,x^{\prime}\in \overline{\mathcal{F}}$, \textsl{and} $\alpha\in (0,1)$,
            \begin{equation*}
            \alpha f + (1-\alpha) x \succsim \alpha g+(1-\alpha) x \Longleftrightarrow \alpha f+(1-\alpha) x^{\prime} \succsim \alpha g+(1-\alpha) x^{\prime}.
            \end{equation*} 
            \end{axiom}

            \indent See \cite{maccheroni2006ambiguity}.\footnote{\cite{de2017rationally} and \cite{ergin2010unique} also discuss this axiom for preferences defined on menus.} Before illustrating this axiom, I introduce the standard independence axiom: 
            
            \noindent \textbf{Axiom I} (\hypertarget{indep}{Independence}). \textsl{For any} $f,g,h\in \mathcal{F}$, \textsl{and} $\alpha\in (0,1)$,
            \begin{equation*}
	        f \succsim g \Longleftrightarrow \alpha f+(1-\alpha) h \succsim \alpha g+(1-\alpha) h.
            \end{equation*} 

            \indent \hyperlink{indep}{Axiom I} is too restrictive in my setting. To see this, let's revisit the \hyperlink{exp1}{investment example}. Except for 
            $f$ and $g$, we have another act $h=(3,1,0)$. Suppose we observe $g\succsim f$. By Axiom I, we expect to observe $\alpha g+(1-\alpha) h \succsim \alpha f+(1-\alpha) h$. However, this is not the case in my model. Let $\alpha = \frac{1}{2}$, then the corresponding partition of $\alpha g+(1-\alpha) h$ is $\{\{\omega_{1}\},\{\omega_{2}\},\{\omega_{3}\}\}$. Similarly, the corresponding partition of $\alpha f+(1-\alpha) h$ is $\{\Omega\}$. It is easy to see that $\alpha f+(1-\alpha) h$ gives higher expected utility than $\alpha g+(1-\alpha) h$ and is less complex than $\alpha g+(1-\alpha) h$. Hence, the agent prefers $\alpha f+(1-\alpha) h$ to $\alpha g+(1-\alpha) h$, which violates Independence. \hyperlink{wci}{Weak Certainty Independence} is more compatible with the agent's behavior in my setting. By the definition of $\sigma(f)$, $\sigma( \alpha f+(1-\alpha) x)$ and $\sigma(f)$ have the same level of complexity for any $f\in\mathcal{F}$ and any $\alpha\in(0,1)$.

            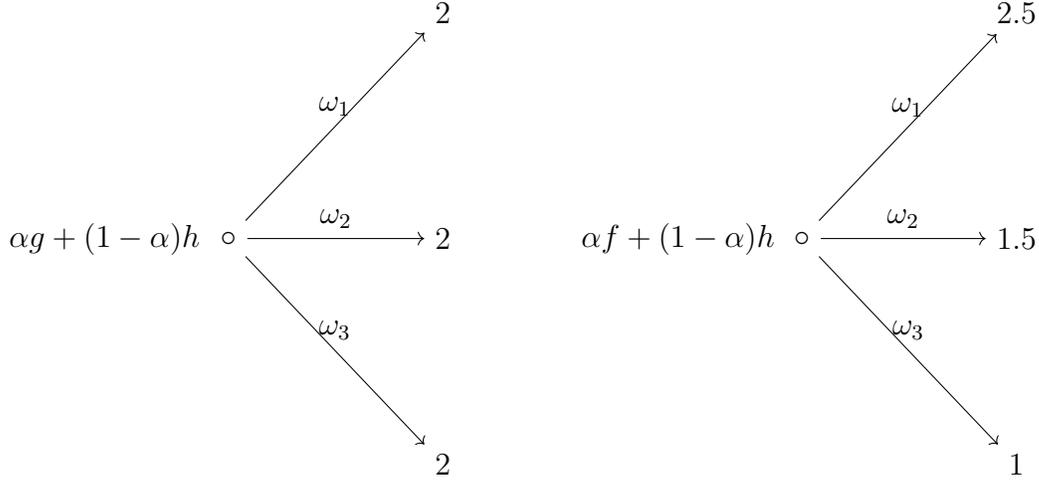
\begin{figure}[ht]
            \centering
            \begin{minipage}{0.49\textwidth} 
            \begin{tikzpicture}[scale=1.5]

            \node (a) at (0,0) {$\alpha g + (1-\alpha) h$};    
            \node (b) at (1.1,0) {$\circ$};
            \node (z1) at (3,2) {$2$};    
            \node (z2) at (3,0) {$2$};    
            \node (z3) at (3,-2) {$2$};

            \draw [->] (b) -- (z1) node[midway,above] {$\omega_{1}$};     
            \draw [->] (b) -- (z2) node[midway,above] {$\omega_{2}$};     
            \draw [->] (b) -- (z3) node[midway,above] {$\omega_{3}$}; 

            \end{tikzpicture}
            \end{minipage}
            \begin{minipage}{0.40\textwidth} 
            \begin{tikzpicture}[scale=1.5]

            \node (a) at (0,0) {$\alpha f + (1-\alpha) h$};  
            \node (b) at (1.1,0) {$\circ$};
            \node (z1) at (3,2) {$2.5$};    
            \node (z2) at (3,0) {$1.5$};    
            \node (z3) at (3,-2) {$1$};

            \draw [->] (b) -- (z1) node[midway,above] {$\omega_{1}$};     
            \draw [->] (b) -- (z2) node[midway,above] {$\omega_{2}$};     
            \draw [->] (b) -- (z3) node[midway,above] {$\omega_{3}$}; 

            \end{tikzpicture}
            \end{minipage}
            \caption{Violation of Independence.}
            \label{fig:figure3}
            \end{figure}
 
            \begin{axiom}[\hypertarget{wm}{Weak Monotonicity}] 
            \textsl{For any} $f,g\in \mathcal{F}$, \textsl{if} $\sigma(f)\subset \sigma(g)\textsl{ and }f(\omega)\succsim g(\omega)$ \textsl{for all} $\omega\in \Omega$, \textsl{then} $f\succsim g$.
            \end{axiom}

            \indent For the monotonicity axiom, I put an additional restriction compared to the traditional definition of monotonicity. To illustrate the idea, consider the following case.  If $f$ is more complex than $g$, the cost of $f$ is greater than $g$, thus, $f(\omega)\succsim g(\omega)$ cannot make sure $f\succsim g$. Thus, I require weak monotonicity.  First, if $f$ and $g$ have the same level of complexity, then $f(\omega)\succsim g(\omega)$ for all $\omega\in \Omega$ directly implies $f\succsim g$.  Second, if $\sigma(f)\subset\sigma(g)$, that is $g$ is more complex than $f$, the cost of $f$ is less than $g$, thus we can apply the weak monotonicity to this case. The violation of monotonicity is documented in many studies such as \cite{birnbaum2008new}, and \cite{gneezy2006uncertainty}). They show that agents prefer a dominated lottery. This phenomenon cannot be explained by standard expected utility theory. In contrast, the behavior of a complexity averse agent may violate monotonicity and can be characterized by this model. 

            \begin{notation} 
            \hypertarget{notation1}{\textsl{For}} \textsl{any act} $g\in\mathcal{F}$, \textsl{let} $\mathcal{F}^{c}(g) = \{f: \sigma(g)\subset \sigma(\lambda f + (1-\lambda)g) \textsl{ for any } \lambda\in(0,1)\}$.  
            \end{notation}

            \begin{axiom}[\hypertarget{ac}{Aversion to Complexity}] 
            \textsl{For any} $\alpha,\beta,\lambda\in(0,1)$ \textsl{with} $\alpha>\beta$, \textsl{any} $x\in \overline{\mathcal{F}}$, $g\in \mathcal{F}$, \textsl{and} $f\in\mathcal{F}^{c}(g)$ \textsl{we have}
            \begin{equation*}
	        \alpha x + (1-\alpha) f \sim \beta x + (1-\beta) f \ \ \Longrightarrow\ \  \lambda x+(1-\lambda) g \succsim \lambda f+(1-\lambda) g.
            \end{equation*} 
            \end{axiom}
            
            The challenge of characterizing the agent's attitude to complex acts is to use appropriate axiom to depict her considerations of the complexity cost of acts. In my framework, since the agent has limited cognitive ability to process complex acts, she has to balance the desire for complex acts and the cost of complexity. For any $\alpha,\beta\in(0,1)$ with $\alpha>\beta$, $\alpha x + (1-\alpha) f \sim \beta x + (1-\beta) f$ implies that the agent is indifferent between $x$ and $f$ if she does not have to consider the complexity cost of $f$. Then, if she has to choose an act between $\lambda x+(1-\lambda) g$ and $\lambda f+(1-\lambda) g$, she will choose the former one since $\lambda f+(1-\lambda) g$ is more complex than $\lambda x+(1-\lambda) g$. However, an SEU agent would be indifferent between these two acts, if she is indifferent between $x$ and $f$. 
            
            To illustrate this axiom, consider the following modified investment example in \hyperref[fig:figure4]{Figure 4}. Now let $f = (1,2,3)$ and $g = (2,2,2)$. It is obvious that acts $f$ and $g$ have the same expected utility given $\mu=(\frac{1}{3},\frac{1}{3},\frac{1}{3})$. In other words, the agent is indifferent between $f$ and $g$ if she does not have to consider the complexity cost of $f$. Suppose she has two more options: $\lambda f + (1-\lambda) h$ and $\lambda g + (1-\lambda) h$, where $h = (4,4,6)$. Again, we  know that $\lambda f + (1-\lambda) h$ has the same expected utility as $\lambda g + (1-\lambda) h$. Moreover, it is easy to see that $\sigma(\lambda f + (1-\lambda) h)$ is finer than $\sigma(g)$ for all $\lambda\in(0,1)$. Thus, if the agent is complexity averse, she prefers $\lambda g + (1-\lambda) h$ to $\lambda f + (1-\lambda) h$.\footnote{By the definition of $\sigma(f)$, $\sigma( \alpha f+(1-\alpha) x)=\sigma(f)$ for any $f\in\mathcal{F}$ and any $\alpha\in(0,1)$.}
\ \\
            \begin{figure}[ht]
            \centering
            \begin{minipage}{0.49\textwidth} 
            \begin{tikzpicture}[scale=1.5]

            \node (a) at (0,0) {$\alpha g + (1-\alpha) h$}; 
            \node (b) at (1.1,0) {$\circ$};
            \node (z1) at (3,2) {$4-2\lambda$};    
            \node (z2) at (3,0) {$4-2\lambda$};    
            \node (z3) at (3,-2) {$6-4\lambda$};

            \draw [->] (b) -- (z1) node[midway,above] {$\omega_{1}$};     
            \draw [->] (b) -- (z2) node[midway,above] {$\omega_{2}$};     
            \draw [->] (b) -- (z3) node[midway,above] {$\omega_{3}$}; 

            \end{tikzpicture}
            \end{minipage}
            \begin{minipage}{0.49\textwidth} 
            \begin{tikzpicture}[scale=1.5]

            \node (a) at (0,0) {$\alpha f + (1-\alpha) h$};   
            \node (b) at (1.1,0) {$\circ$};
            \node (z1) at (3,2) {$4-3\lambda$};    
            \node (z2) at (3,0) {$4-2\lambda$};    
            \node (z3) at (3,-2) {$6-3\lambda$};

            \draw [->] (b) -- (z1) node[midway,above] {$\omega_{1}$};     
            \draw [->] (b) -- (z2) node[midway,above] {$\omega_{2}$};     
            \draw [->] (b) -- (z3) node[midway,above] {$\omega_{3}$}; 

            \end{tikzpicture}
            \end{minipage}
            \caption{An Example illustrating Aversion to Complexity.}
            \label{fig:figure4}
            \end{figure}
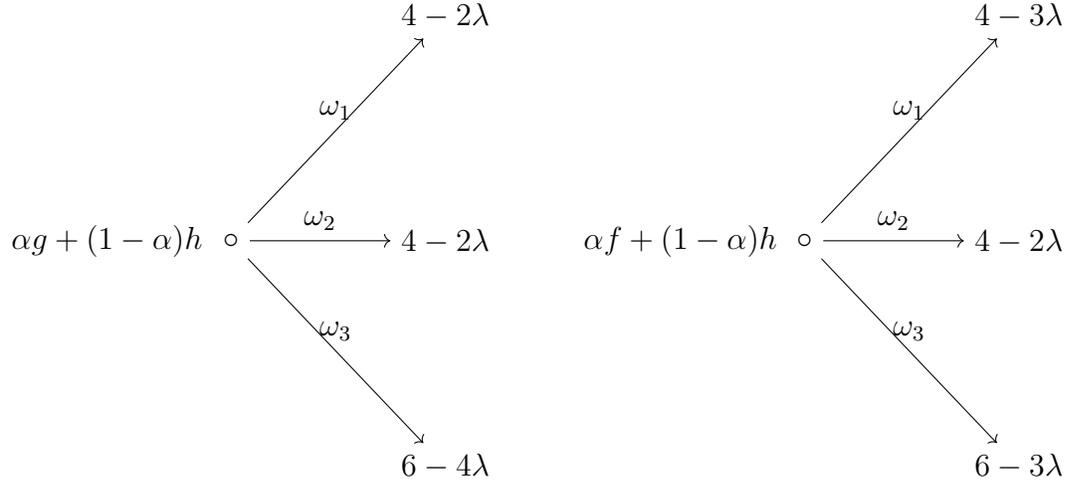

            \begin{axiom}[\hypertarget{unb}{Unboundedness}] 
            \textsl{There exist outcomes} $x$ \textsl{and} $y$, \textsl{with} $x\succ y$, \textsl{such that for any} $\alpha\in (0,1)$, \textsl{there exists an outcome} $z$ \textsl{such that either} $y\succ \alpha z+(1-\alpha)x$ \textsl{or} $\alpha z+(1-\alpha)y\succ x$.
            \end{axiom}

            \indent \hyperlink{unb}{Unboundedness} axiom implies that preferences over outcomes are unbounded (the range of $u(\cdot)$ is $\mathbb{R}$). $y\succ \alpha z+(1-\alpha)x$ implies $u(y)\geq u( \alpha z+(1-\alpha)x)=\alpha u(z)+(1-\alpha)u(x)$, and $\alpha z+(1-\alpha)y\succ x$ implies $u( \alpha z+(1-\alpha)y)=\alpha u(z)+(1-\alpha)u(y)\geq u(x)$. When $\alpha$ is very close to zero, to keep $x\succ y$, $u(z)$ should be $-\infty$ or $\infty$.  It is used to prove \hyperlink{lem2}{Lemma 2}.

            \hyperlink{continuity}{Axiom 2-4} are the assumptions of continuity, independence, and monotonicity, which are the classic conditions paving the way to the subjective expected utility (see \cite{anscombe1963definition}). Note that \hyperlink{wm}{Axiom 4} and \hyperlink{ac}{5} require that the agent is able to have a full understanding of acts as functions that attach outcomes to states. With this assumption, the agent is able to identify the complexity of all acts. With \hyperlink{wo}{Axiom 1-6}, we define the preference on acts in this model. 

\newpage
\subsection{The Representation} \label{sec3.2}

            \begin{definition}[\hypertarget{def1}{Complexity Aversion Representation}]  
            \textsl{An agent admits a Complexity Aversion Representation if there exist}
            \begin{itemize}
	        \item{\textsl{an unbounded affine utility function} $u: X \rightarrow \mathbb{R}$};
	        \item{\textsl{a probability measure} $\mu\in \Delta(\Omega)$}; \textsl{and}
	        \item{\textsl{a complexity cost function} $\mathcal{C}:\{\sigma(S): S\in\mathbb{P}\} \rightarrow [0, \infty)$, $\sigma(S)\subset \sigma(S^{\prime})$ \textsl{implies} $\mathcal{C}(\sigma(S))\leq \mathcal{C}(\sigma(S^{\prime})$}, \textsl{and} $\mathcal{C}(\{\Omega,\emptyset\})=0$.
            \end{itemize}
            \noindent \textsl{Such that}: $\succsim$ \textsl{is represented by} 
            \begin{equation*}
            V(f) = \int_{\omega \in \Omega} u(f(\omega))\mu(d\omega)-\mathcal{C}(\sigma(f)).
            \end{equation*}
            \end{definition}

            I refer to the \hyperlink{def1}{Complexity Aversion Representation} as a tuple $\langle u, \mu, \mathcal{C}\rangle$. The utility function $u$ and belief $\mu$ have the same interpretations as in the model of SEU. The complexity cost function $\mathcal{C}$ maps each act to an extended real number measuring the cost that the agent has to bear if she chooses that act. If the $\sigma$-algebra generated by the corresponding partition of act $f$ is finer than that of act $g$, then $f$ has higher complexity cost than $g$.

            \indent I am now ready to state the main result : the representation theorem of complexity aversion.

            \begin{theorem} 
            \hypertarget{thm1}{A} preference relation $\succsim$ satisfies \hyperlink{wo}{Axioms 1-6} if and only if  it admits a \hyperlink{def1}{Complexity Aversion Representation}. 
            \end{theorem}

            \hyperlink{thm1}{Theorem 1} provides a behavioral foundation for the Complexity Aversion representation. A noteworthy remark on \hyperlink{thm1}{Theorem 1} before proceeding. I do not assume any other properties about the  complexity cost function except for monotonicity. Some other properties can be imposed on the complexity cost function like convexity. By doing this, appropriate behavioral axioms should be characterized. 

            \indent The proof is in \hyperref[appA.1]{Appendix A.1}. Here I only present a sketch and some discussions of the proof. First, it is obvious that if the expected utility of act $f$ is smaller than that all constant acts, then the agent definitely will not prefer $f$ to $x$ because $x$ is the simplest act that she can choose. Thus, the agent only considers acts that might be preferred to some constant acts. Only for those kind of acts, complexity cost is meaningful. By this observation, I can prove that there exists certainty equivalence for such acts. Then \hyperlink{wci}{Axiom 3} and \hyperlink{ac}{Axiom 5} shows that there exist an affine utility function $uR$ with unbounded range and a prior probability measure $\mu$ over $\Omega$ such that $U(f)=\int_{\Omega} u (f(\omega)) \mu(d\omega)$. 

            \indent Second, I prove the existence of such complexity cost function $\mathcal{C}$ by construction. We can observe that $f\sim g$ does not mean that the two acts have the same expected utility. Then we can construct an act $g$ by mixing a constant act and act $f^{\prime}$, where $u(f^{\prime}(\omega))=2u(f(\omega))$. By the definition of $\sigma$-algebra, we observe that $\sigma(f)=\sigma(g)$. Therefore, by rearranging the expected utility of $g$, we construct such complexity cost function $\mathcal{C}$ for act $f$.

            \indent The last step is to show that the monotonicity of complexity cost function $\mathcal{C}$. Suppose $\alpha x + (1-\alpha) f \sim \beta x + (1-\beta)f$ for any $\alpha,\beta\in(0,1)$. It is easy to see that $x$ and $f$ have the same expected utility. Then we observe that $\lambda x + (1-\lambda)g$ and $\lambda f + (1-\lambda)g$ also have the same expected utility. By Axiom 2.5, $\lambda x + (1-\lambda)g\succsim\lambda f + (1-\lambda)g$ if $\sigma(g)\subset \sigma(\lambda f + (1-\lambda)g)$ implies $\mathcal{C}(\sigma(g))\leq \mathcal{C}(\sigma(\lambda f + (1-\lambda)g))$. 

\subsection{Identification} \label{sec3.3}

            \noindent In this subsection, I discuss the uniqueness of a \hyperlink{def1}{Complexity Aversion Representation}. Moreover, I give an example to illustrate the identification of complexity cost function $\mathcal{C}$.

            \indent At a glance, the cost function $\mathcal{C}$ seems not be unique. For example, one may say that $\mathcal{C}+c$ ($c\in \mathbb{R}$ is a constant) might represent the same preference relations if it does not change the ordinal ranking over acts. Actually, if $\mathcal{C}$ is defined as in \hyperlink{def1}{Definition 1}, it is unique. To see this, for any act $\bar{f}\in\mathcal{F}$, consider a modified act $f^{\prime}$ such that $f^{\prime}(\omega)=\bar{f}(\omega)-\epsilon$ with a very small but positive $\epsilon$ for all $\omega\in\Omega$. Suppose we can find another complexity cost function $\mathcal{C}^{\prime}$ that represents the same preference. The only difference between $\mathcal{C}$ and $\mathcal{C}^{\prime}$ is that $\mathcal{C}^{\prime}\big(\sigma(f)\big)>\mathcal{C}\big(\sigma(f)\big)$ for any act $f\in\mathcal{F}^{S^{\bar{f}}}$. Since $\sigma(\bar{f})=\sigma(f^{\prime})$, we have $\mathcal{C}\big(\sigma(\bar{f})\big)=\mathcal{C}\big(\sigma(f^{\prime})\big)$. Moreover, suppose $f^{\prime}\sim g$ with complexity cost function $\mathcal{C}$. It is easy to see that  $\bar{f}\succ g$ with complexity cost function $\mathcal{C}$. Then it is impossible to have  $f^{\prime}\sim g$ with complexity cost function $\mathcal{C}^{\prime}$. In this case, two complexity cost functions are only different on partition $S^{f}$, If two complexity cost functions vary considerably, the preference relations will also vary considerably. Therefore, we have the following uniqueness corollary of complexity cost function $\mathcal{C}$.

            \begin{corollary} 
            \hypertarget{coro1}{Let} $\succsim$ be a complexity aversion preference represented by $\langle u,\mu, \mathcal{C}\rangle$. Then the complexity cost function, defined in \hyperlink{def1}{Definition 1}, is unique.
            \end{corollary}

            \begin{corollary} 
            \hypertarget{coro2}{If} $\langle u,\mu,\mathcal{C}\rangle$  and $\langle u^{\prime},\mu^{\prime},\mathcal{C}^{\prime}\rangle$ represent the same preferences relations, then $u^{\prime}$ is a positive affine transformation of $u$, $\mu=\mu^{\prime}$ and $\mathcal{C}^{\prime}=\alpha \mathcal{C}$ for some $\alpha>0$.
            \end{corollary}

            \indent \hyperlink{coro2}{Corollary 2} establishes that the agent's utility function, prior, and complexity cost function are unique. It is a standard practice to identify the agent's utility. To identify the cost function, consider two acts $f\in\mathcal{F}$ and $z\in\overline{\mathcal{F}}$ such that $u\big(f(\omega)\big)=u(x)>u(z)$ for all $\omega\in s_{1}$, $u\big(f(\omega)\big)=u(y)<u(z)$ for all $\omega\in s_{2}$ and $\{s_{1},s_{2}\} \in\mathbb{P}$. Then the agent prefers $f$ to $z$ whenever $u(f)-\mathcal{C}\big(\sigma(f)\big)> u(z) - \mathcal{C}\big(\sigma(z)\big)$. Suppose $u(y)$ is very close to $0$ that is $y$ is a significantly unpreferred outcome. Then she prefers $f$ to $z$ if $u(x)-\mathcal{C}\big(\sigma(f)\big)> u(z) - \mathcal{C}\big(\sigma(z)\big)$. Therefore, we can identify $\mathcal{C}\big(\sigma(f)\big)$ uniquely by $min\{u(x)-u(z): x\in\overline{\mathcal{F}}\}$ given $z$.  

\subsection{Comparative Statics} \label{sec3.4}
\subsubsection{Comparing the Degree of Complexity Aversion} \label{sec3.4.1}

            For the preference $\succsim$, the agent faces a trade-off between her desire for complex acts and their costs. To formalize the notion, let $\langle u^{1},\mu^{1},\mathcal{C}^{1}\rangle$ and $\langle u^{2},\mu^{2},\mathcal{C}^{2}\rangle$ represent two agents' preferences $\succsim^{1}$ and $\succsim^{2}$. To compare the cost function, let $(u^{1}, \mu^{1})=(u^{2},\mu^{2})$.

            \begin{definition}  
            \hypertarget{def2}{$\succsim^{1}$} \textsl{has a} \textbf{lower degree of complexity aversion} \textsl{than} $\succsim^{2}$ \textsl{if for any} $f\in \mathcal{F}$ \textsl{and} $x\in  \overline{\mathcal{F}}$, $x \succ^{1} f$ \textsl{implies} $x \succ^{2} f$.
            \end{definition}

            \indent The definition says that the DM1 has a lower degree of complexity aversion than DM2 if: DM1 strictly prefers a constant act $x$ to an act $f$ implies DM2 also strictly prefers $x$ than $f$. The following result shows that DM1 has a lower degree of complexity aversion than DM2 in terms of the parameters of the Complexity Aversion representation.  

            \begin{theorem}  
            \hypertarget{thm2}{$\succsim^{1}$} has \textsl{a lower degree of complexity aversion} than $\succsim^{2}$ if and only if DM1 has lower complexity costs than DM2, that is $\mathcal{C}^{1}\leq \mathcal{C}^{2}$.
            \end{theorem}

            \indent \hyperlink{thm2}{Theorem 2} characterizes that if DM1 has a lower degree of complexity aversion than DM2, then her complexity cost $\mathcal{C}^{1}$ is smaller than $\mathcal{C}^{2}$. In other words, DM1 needs less efforts to process every act than DM2. 

\subsubsection{Comparing the Capacity for Complex Acts} \label{sec3.4.2}

            The second comparison considers two agents' capacities for complex acts. For the preference $\succsim$, except for the trade-off between her desire for complex acts and their costs, the agent also faces the trade-off between her desire for complex acts and her aversion to randomized acts. Consider two randomized acts $\alpha f + (1-\alpha) g$ and $\alpha f + (1-\alpha) x$. It is easy to understand that act $g$ is more complex than act $x$, while the former is more random than the latter. Therefore, the DM1 has a higher capacity for complex acts than DM2 if DM1 is more capable of exploring complex acts than DM2. 

            \begin{definition}  
            \hypertarget{def3}{$\succsim^{1}$} \textsl{has a} \textbf{higher capacity for complex acts} \textsl{than} $\succsim^{2}$ \textsl{if for any} $f, g\in \mathcal{F}$, $x\in  \overline{\mathcal{F}}$, \textsl{and} $\alpha\in (0,1)$, $\alpha f + (1-\alpha) x\succ^{1} \alpha f + (1-\alpha) g$ \textsl{implies} $\alpha f + (1-\alpha) x\succ^{2} \alpha f + (1-\alpha) g$.
            \end{definition}

            \indent This definition says that the DM1 has a higher capacity for complex acts than DM2 if: DM1 strictly prefers  $\alpha f + (1-\alpha) x$ to $\alpha f + (1-\alpha) g$ implies DM2 also strictly prefers $\alpha f + (1-\alpha) x$ to $\alpha f + (1-\alpha) g$. The following result characterizes that DM1 has a higher capacity for complex acts  than DM2 in terms of the parameters of the \hyperlink{def1}{Complexity Aversion Representation}.

            \begin{theorem} 
            \hypertarget{thm3}{$\succsim^{1}$} has \textsl{a higher capacity for complex acts} than $\succsim^{2}$ if and only if $supp(\mathcal{C}^{1})\subset supp(\mathcal{C}^{2})$.
            \end{theorem}

            \indent $supp(\mathcal{C})$ is the set of acts that the agent would choose.  $supp(\mathcal{C}^{1})\subset supp(\mathcal{C}^{2})$ indicates that DM1 has a a higher capacity to exclude coarser acts than DM2.

\section{Minimal Complexity Updating} \label{sec4}

            In this part, I discuss a possible ``updating rule" that an agent could use when she receives the new information. To avoid confusion, I have to point out that I am not discussing belief ($\mu\in\Delta(\Omega)$) updating. Instead, I focus on what kinds of conditional cost function she will use to measure the complexity cost of an act when $E$ occurs.

\subsection{Axioms} \label{sec4.1}

            Recall the standard definition of a \textsl{null} event: for any preference relation $\succsim$, an event $E\subseteq \Omega$ is called $\succsim\text{-null}$ if $fEg\sim g$ for any acts $f,g\in\mathcal{F}$. In the expected utility framework, null events have zero probability. Since the way of a complexity averse agent ranking acts is different from an SEU agent, here I present a new version of the definition of a \textsl{null} event: for any preference relation $\succsim$, an event $E\subseteq \Omega$ is called $\succsim\text{-null}$ if there exist no $f\in\mathcal{F}$ and $x,y\in\overline{\mathcal{F}}$ such that $xEf\succ yEf$. This definition is more compatible with my setting because constant acts $x$ and $y$ have the same level of complexity on $E$.  

            The following three axioms impose dynamic properties on the agent's preference relations which are related to how the agent processes the non-null events in $\Sigma^{\prime}$. First, I introduce \hyperlink{conse}{Consequentialism}. In the SEU setting, \hyperlink{conse}{Consequentialism} is the axiom that guarantees that the preference conditional on $E\in  \Sigma^{\prime}$ does not depend on how act $f$ behaves outside of $E$. In other words, the agent believes that the true state must lie in $E$ and she is indifferent between two acts that differ only outside of $E$. 

            \begin{axiom}[\hypertarget{conse}{Consequentialism}] \textsl{For any} $E\in \Sigma^{\prime}$, \textsl{and any} $f,g\in \mathcal{F}$, \textsl{if} $f(\omega)=g(\omega)$ \textsl{for all} $\omega\in E$, \textsl{then} $f\sim_{E} g$.
            \end{axiom}

            \indent To some extend, this standard axiom is enough to characterize the behavior of agents in this model. To see this, consider two acts $f$ and $g$ such that $f$ is more complex than $g$ on $\Omega$. Since $f(\omega)=g(\omega)$ for all $\omega\in E$, we have $\sigma(f|E) = \sigma(g|E)$. Thus, if the agent measures the cost of $f$ and $g$ on $E$ instead of $\Omega$ after the occurrence of $E$, that is, the conditional complexity cost of $f$ is the same as that of $g$, then she is indifferent between $f$ and $g$ conditional on the realization of $E$. I discuss the conditional complexity cost function below. 

            \indent Before that, I turn to discuss the relations between ex ante preferences and ex post preferences. In the standard setting, the preference is required to be dynamically consistent. 

            \noindent \textbf{Axiom DC} (\hypertarget{dc}{Dynamic Consistency}). \textsl{For any} $E\in \Sigma^{\prime}$, \textsl{and any} $f,g\in \mathcal{F}$, \textsl{we have} $f\succsim_{E} g\Longleftrightarrow fEg\succeq g$.  

            \indent \hyperlink{dc}{Dynamic Consistency} requires that the ex ante preference over acts implies the agent's ex post preference. In particular, if the agent prefers $f$ in $E$ to $g$ before the arrival of new information, then if $E$ happens, the agent still prefers $f$ to $g$.\footnote{\cite{ghirardato2002revisiting}) provides more discussion of this axiom and its implications. He proves that \hyperlink{conse}{Consequentialism} and \hyperlink{dc}{Dynamic Consistency} imply that the agent is a Bayesian updater.}  However,  if the agent is complexity averse, things are different. For example, consider two acts $f$ and $g$ in $\mathcal{F}$, such that $\int_{\omega\in\Omega} u(f(\omega))\mu(d\omega)\leq\int_{\omega\in\Omega} u(g(\omega))\mu(d\omega)$ and $ \sigma(f)\subset \sigma(g)$. Suppose she prefers $f$ to $g$, which means although $g$ gives higher expected utility than $f$,  the complexity cost of act $g$ is too high. Then, after the arrival of new information, it is possible that $g$ still gives higher expected utility than $f$ on $E$, but with lower complexity cost on $E$. Under this circumstance, she prefers $g$ to $f$ after she receives the new information.  Thus,  standard  dynamic consistency will not apply. Before I introduce the axiom, consider the following definition of conditional complexity cost function. 

            \begin{definition} 
            \hypertarget{def4}{\textsl{Given}} $E\in\Sigma^{\prime}$ \textsl{and} $\mu\in \Delta(\Omega)$, \textsl{a conditional complexity cost function} $\mathcal{C}_{E,\mu}: \{\sigma(f): f\in\mathcal{F}\}\longrightarrow \mathbb{R}_{+}\cup \{\infty\}$ \textsl{is} \textbf{minimal} if $\mathcal{C}_{E,\mu}\big(\sigma(f)\big)=min\{\mathcal{C}\big(\sigma(h)\big)/\mu(E): h\in \mathcal{F} ,\textsl{ and }\sigma(h|E)=\sigma(f|E)$\}.\footnote{Here, $\mu(E)$ is a technique term that plays a role in the proof of \hyperlink{thm2}{Theorem 2}.} 
            \end{definition}

            \indent By \hyperlink{conse}{Axiom 7}, we know that when the agent receives the new information $E\in\Sigma^{\prime}$, she believes that the true state must lie inside $E$. Then she does not have to care about the complexity of an act outside of $E$. Therefore, after the arrival of $E$, she can choose an act that gives the same $\sigma(f|E)$ as $f$ on $E$ but with lowest costs on $\Omega$, and use $\mathcal{C}\big(\sigma(h)\big)$ to measure the cost of $\mathcal{C}_{E,\mu}(\sigma(f))$.

\begin{figure}[!ht]
            \centering
	        \begin{tikzpicture}
		    \tikzset{
			mydot/.style={
				fill,
				circle,
				inner sep=1.5pt
			    }
		    }
		    \coordinate  (A) at(0,0) ;
		    \coordinate  (B)  at (2,0);
		    \coordinate  (C) at (4,0);
		    \coordinate  (D) at (8,0);
		    \coordinate  (E) at (10,0);
		    \draw (A) -- node [midway, below]{\scriptsize $\{\omega_{1},\ \omega_{2}\}$}
		    (B)  -- node [midway, below]{\scriptsize $\{\omega_{3},\ \omega_{4}\}$}
		    (C)   -- node [midway, below]{\scriptsize $\{\omega_{5},\ \ \omega_{6},\ \ \omega_{7},\ \ \omega_{8}\}$}
		    (D)  -- node [midway, below]{\scriptsize $\{\omega_{9},\ \omega_{10}\}$}
		    (E);
		    \node[mydot,label={below left:}] at (A) {};
		    \node[mydot,label={below left:}] at (B) {};
		    \node[mydot,label={below left:}] at (C) {};
		    \node[mydot,label={below left:}] at (D) {};
		    \node[mydot,label={below right:}] at (E) {};
		    \draw[{Bracket[width=27mm]}-,ultra thick,red]  ($(3,0)!-3.5mm!-90:(3.01,0)$)--node {} ($(3.01,0)!-3.5mm!90:(3,0)$);
		    \draw[-{Bracket[width=27mm]},ultra thick,red] ($(6.99,0)!-3.5mm!-90:(7,0)$)--node {} ($(7,0)!-3.5mm!90:(6.99,0)$);
		    \draw (5,0) node[above=12.5mm] {$\color{red}E$};
		    \draw (10.3,0) node[above=-4mm] {$\Omega$};
		    \draw (10.3,0) node[above=12mm] {$f$};
		    \draw (10.3,0) node[above=7mm] {\color{orange}$h$};
		    \draw (1,0) node[above=0.1mm] {$1$};
		    \draw (3.5,0) node[above=4.8mm] {$2$};
		    \draw (6.7,0) node[above=9.8mm] {$3$};
		    \draw (2,0) node[above=9.8mm] {$3$};
		    \draw (9,0) node[above=13.8mm] {$4$};
		    \draw[-,black,line width=2.8pt] ($(A)!0mm!-90:(B)$)--node {} ($(B)!0mm!90:(A)$);
		    \draw[-,orange,line width=4pt] ($(A)!-10mm!-90:(3,0)$)--node {} ($(3,0)!-10mm!90:(A)$);
		    \draw[-,orange,line width=6pt] ($(3,0)!-5.1mm!-90:(4,0)$)--node {} ($(4,0)!-5.1mm!90:(3,0)$);
		    \draw[-,orange,line width=5pt] ($(4,0)!-10mm!-90:(10,0)$)--node {} ($(10,0)!-10mm!90:(4,0)$);
		    \draw[-,black,line width=2.8pt] ($(B)!-5mm!-90:(C)$)--node {} ($(C)!-5mm!90:(B)$);
		    \draw[-,black,line width=2.7pt] ($(C)!-10mm!-90:(D)$)--node {} ($(D)!-10mm!90:(C)$);
		    \draw[-,black,line width=3pt] ($(D)!-14mm!-90:(E)$)--node {} ($(E)!-14mm!90:(D)$);
	        \end{tikzpicture}
            \caption{Example of Minimal Complexity Updating.}
            \label{fig:figure5}
            \end{figure}
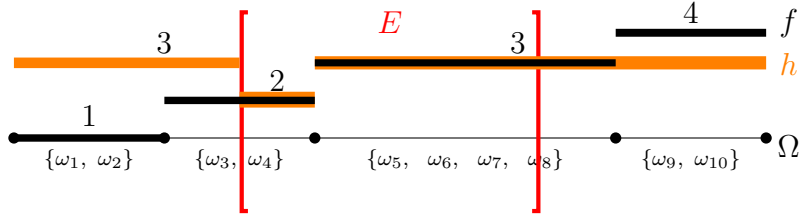

            \begin{example}
                \hypertarget{exp2}{Consider} the following example in \hyperref[fig:figure5]{Figure 5}. Act $f$  gives $4$ distinct outcomes (represented by black dashes) for all states in $\Omega$. Now suppose event $E$ is realized, so we can see that $f$ gives $2$ distinct outcomes for states in $E$. Let $S^{f|E}$ denote the partition generated by $f$ on $E$, i.e., $S^{f|E} = \{\{\omega_{4},\{\omega_{5},\omega_{6},\omega_{7}\}\}\}$. It is obvious that element $\{\omega_{5},\omega_{6},\omega_{7}\}$ has more states than any any other elements in $E$ (top right black dash in $E$). If the agent uses a minimal conditional complexity cost function, then she would choose $h=fEx$ where $x\in f(E)$ and $f(\omega)=x$ for all $\omega\in \{\omega_{5},\omega_{6},\omega_{7}\}$ (represented by orange dashes) and use use $\mathcal{C}\big(\sigma(h)\big)$ to measure the cost of $\mathcal{C}_{E}(\sigma(f))$. 
            \end{example}

            \begin{notation} \hypertarget{notation2}{\textsl{For}} \textsl{any act} $f\in\mathcal{F}$ \textsl{and} $E\in\Sigma^{\prime}$, \textsl{let} $\mathcal{F}^{min}(f,E) = \big\{h: h\in \mathcal{F}(f,E)\textsl{ and }\\ \sigma(h)\subset \sigma(h^{\prime}) \textsl{ for all } h^{\prime}\in \mathcal{F}(f,E)\}$, \textsl{where} $\mathcal{F}(f,E) = \{h: h\in\mathcal{F} \textsl{ and } \sigma(h|E)=\sigma(f|E)\}$. 
            \end{notation}

            \indent The following axiom describes the behavior of a complexity averse agent after the arrival of new information.

            \begin{axiom} [\hypertarget{mcu}{Minimal Complexity Updating}] \textsl{Given any} $E\in \Sigma^{\prime}$, $f\in \mathcal{F}$, \textsl{and} $x,z\in \overline{\mathcal{F}}$. \textsl{For any} $z^{\prime}\in \overline{\mathcal{F}}$ \textsl{such that} $\sigma(fEz)\subset \sigma(fEz^{\prime})$, \textsl{we have} 
            \begin{equation*}
            fEz\succsim_{E} x\Longrightarrow fEz^{\prime}\succsim_{E} x.
            \end{equation*}
            \end{axiom}

            \indent This axiom concerns the agent's attitude toward the complexity cost of an act $f$ outside of event $E$. To see the behavioral implications of minimal conditional complexity cost function, consider an act $f\in\mathcal{F}$. If the agent believes that the true state must lies in $E$, she is indifferent between $fEz$ and $fEz^{\prime}$ where $\sigma(fEz|E)=\sigma(fEz^{\prime}|E)$ for any $z,z^{\prime}\in \overline{\mathcal{F}}$. As such, if we observe the preference $fEz\succsim_{E} x$ where $x\in \overline{\mathcal{F}}$, then we would expect that  a more complex (on $\Omega$) act $fEz^{\prime}$ is still preferred to $x$. 

            \indent I am now ready to be back to discuss the violation of \hyperlink{dc}{Dynamic Consistency}. To illustrate this behavioral postulate, consider the following two cases for $f,g\in\mathcal{F}$ and $x\in\overline{\mathcal{F}}$:
            \begin{equation*}
	        \begin{split}
	        & case1: f\succsim g \text{ and } fEx\succsim gEx, \\
	        & case2:  g\succ f \text{ and } fEx\succsim gEx.
	        \end{split}
            \end{equation*}

            \noindent In case 1, the preference between $f$ and $g$ is consistent with the preference $f$ and $g$ on $E$. However, in our model, case 1 requires more specific structures of cost function, e.g., $\mathcal{C}\big(\sigma(g)\big) - \mathcal{C}\big(\sigma(f)\big) = \mathcal{C}\big(\sigma(gEx)\big) - \mathcal{C}\big(\sigma(fEx)\big)$. Case 2 violates \hyperlink{dc}{Dynamic Consistency} (e.g., \hyperlink{exp1}{the investment example}). Even with a minimal conditional complexity cost function, we cannot exclude case 2 in this model. Instead of \hyperlink{dc}{Dynamic Consistency}, I introduce an axiom called \hyperlink{dca}{Dynamic Complexity Aversion} that aligns with an agent's minimal conditional complexity cost function. 

            \begin{axiom}[\hypertarget{dca}{Dynamic Complexity Aversion}] 
            \textsl{For any} $E\in \Sigma^{\prime}$, $f,g\in \mathcal{F}$, \textsl{and} $h\in \mathcal{F}^{min}(g,E)$, \textsl{there exists} $x\in g(E)$ \textsl{such that} $h=gEx$. \textsl{We have}:
            \begin{equation*}
            fEx\succsim gEx \Longrightarrow f\succsim_{E} g.
            \end{equation*}
            \end{axiom}

            \indent Here I restrict $gEx$ in  $\mathcal{F}^{min}(g,E)$, in words, $gEx$ is the simplest act that gives the same $\sigma(g|E)$ as $g$ on $E$. If we observe the preference $fEx\succsim gEx$, by axiom A8, the agent will use $\mathcal{C}\big(\sigma(fEx_{f})\big)$ $\big(\text{ where }fEx_{f}\in \mathcal{F}^{min}(f,E)\big)$ to measure the cost of $\mathcal{C}_{E,\mu}(\sigma(f))$ when $E$ is realized. Since $\sigma(fEx_{f})\subset \sigma(fEx)$, $fEx$ must be more costly than $fEx_{f}$, we must have that $f\succsim_{E}g$. \hyperlink{dca}{Dynamic Complexity Aversion} can be viewed as an introspective reaction when the agent is self-aware of the updating of complexity cost function.\footnote{\hyperlink{dca}{Dynamic Complexity Aversion} indicates that the agent is introspective. She knows that a new complexity cost function will be used after the arrival of $E$. Thus, she pretends that she will measure the complexity cost of an act by using conditional complexity cost function at ex ante. If she finds that she prefers $fEx$ to $gEx$, then she learns that she will still prefer $f$ to $g$ after the occurrence of $E$.}  Broadly speaking, \hyperlink{dca}{Dynamic Complexity Aversion} is neither weaker nor stronger than  \hyperlink{dc}{Dynamic Consistency}. Although it allows for some violations of  \hyperlink{dc}{Dynamic Consistency}, it restricts agents' behavior according to \hyperlink{mcu}{minimal complexity updating}. 

\subsection{The Representation} \label{sec4.2}

            \begin{definition}[\hypertarget{def5}{\textbf{Minimal Complexity Aversion Representation}}]  
            \textsl{An agent admits a Minimal Complexity Aversion Representation if there exist}

            \begin{itemize}
	        \item{\textsl{an unbounded affine utility function} $u: X \rightarrow \mathbb{R}$};
	        \item{\textsl{a probability measure} $\mu\in \Delta(\Omega)$}; \textsl{and}
	        \item{\textsl{a complexity cost function} $\mathcal{C}:\{\sigma(S): S\in\mathbb{P}\} \rightarrow [0, \infty)$, $\sigma(S)\subset \sigma(S^{\prime})$ implies $\mathcal{C}(\sigma(S))\leq \mathcal{C}(\sigma(S^{\prime})$}, \textsl{and} $\mathcal{C}(\{\Omega,\emptyset\})=0$;
	        \item {\textsl{a minimal conditional complexity cost function} $\mathcal{C}_{E,\mu}$} \textsl{defined in \hyperlink{def4}{Definition 4}}.
            \end{itemize}
            \noindent \textsl{Such that}: 
            \begin{enumerate}[label=(\roman*),noitemsep]
            \item $\succsim$ \textsl{is represented by} 
            \begin{equation*}
	        V(f) = \int_{\omega \in \Omega} u\big(f(\omega)\big)\mu(d\omega)-\mathcal{C}\big(\sigma(f)\big) ;
            \end{equation*}
            \item \textsl{For any} $E\in \Sigma^{\prime}$, $\succsim_{E}$ \textsl{is represented by} 
            \begin{equation*}
	        V(f|E) = \int_{\omega \in \Omega} u(f(\omega))\mu_{E}(d\omega)-\mathcal{C}_{E,\mu}\big(\sigma(f)\big).
            \end{equation*}
            \end{enumerate}
            \end{definition}

            \indent  $\langle u, \mu, \mathcal{C}, \{\mathcal{C}_{E,\mu}\}_{E\in\Sigma^{\prime}}\rangle$ is referred to as the \hyperlink{def5}{Minimal Complexity Aversion Representation}. The utility function $u$, belief $\mu$, and the complexity cost function $\mathcal{C}$ have the same interpretations as in the model of \hyperlink{def1}{CAR}. The conditional complexity cost function \hyperlink{def4}{$\mathcal{C}_{E,\mu}$} measures the cost of acts that the agent has to bear after the occurrence of $E$. If the $\sigma$-algebra generated by the corresponding partition of act $f$ conditional on $E$ is finer than that of act $g$, then $f$ has higher conditional complexity cost than $g$. Above three axiom, along with \hyperlink{car}{conditional CAR axioms} that are discussed in \hyperref[sec3.1]{Section 3.1}, leads to the following representation theorem. 

            \begin{theorem} 
            \hypertarget{thm4}{An} ex ante preference relation $\succsim$  and a collection $\{\succsim_{E}\}_{E\in\Sigma^{\prime}}$ of conditional preference relations jointly satisfy \hyperlink{car}{Axioms 1-9} if and only if  they admit a \hyperlink{def5}{Minimal Complexity Aversion Representation}.
            \end{theorem}

            \indent \hyperlink{thm4}{Theorem 4} provides a behavioral foundation for the \hyperlink{def5}{Minimal Complexity Aversion Representation}. Note that \hyperlink{mcu}{minimal complexity updating} does not make further assumptions on the agent's understanding of acts and uncertainty. The standard axiom \hyperlink{conse}{Consequentialism} has already characterized the agent's attitude to the new information $E$.

\subsection{Uniqueness} \label{sec4.3}

            I show that the complexity cost function $\mathcal{C}$ is unique in \hyperref[sec3.3]{Section 3.3}. Furthermore, the following corollary shows that the conditional complexity cost function $\mathcal{C}_{E,\mu}$ shares the same properties as $\mathcal{C}$ and is also unique.

            \begin{corollary} 
            \hypertarget{coro3}{Let} $(\succsim,\succsim_{E})_{E\in\Sigma^{\prime}}$ be a minimal complexity aversion preference represented by $\langle u,\mu,\mathcal{C},\mathcal{C}_{E,\mu}\rangle$. The complexity cost function $\mathcal{C}$ satisfies $(\romannumeral1)$ $\mathcal{C}(\{\Omega,\emptyset\}) = 0$; $(\romannumeral2)$ monotonicity.  Then $\mathcal{C}_{E,\mu}$, defined as in \hyperlink{def4}{Definition 4}, also satisfies above two properties and is unique.
            \end{corollary}

            \begin{corollary} 
            \hypertarget{coro4}{If} $\langle u,\mu,\mathcal{C},\{\mathcal{C}_{E,\mu}\}_{E\in\Sigma^{\prime}}\rangle$  and $\langle u^{\prime},\mu^{\prime},\mathcal{C}^{\prime},\{\mathcal{C}^{\prime}_{E,\mu^{\prime}}\}_{E\in\Sigma^{\prime}} \rangle$ represent the same preferences relations, then $u^{\prime}$ is a positive affine transformation of $u$, $\mu=\mu^{\prime}$, $\mathcal{C}^{\prime}=\alpha \mathcal{C}$ and $\mathcal{C}_{E,\mu^{\prime}}^{\prime}=\alpha \mathcal{C}_{E,\mu}$ for some $\alpha>0$.
            \end{corollary}

            \indent \hyperlink{coro4}{Corollary 4} shows that the agent's utility function, prior, the complexity cost function, and the conditional complexity cost function are unique. 

\section{Applications}  \label{sec5}
\subsection{Design of Contracts Under Moral Hazard}  \label{sec5.1}

            In this section, I discuss an application to optimal contract design of wage scheme. I show that complexity aversion may profoundly change the principals' choice of optimal wage scheme.

            \indent The risk neutral  principal tends to design an incentive contract to hire a manager (the agent) for a specific project. The agent is risk averse and receives a utility $u(W)$ given wage $W\geq 0$ net of the cost $c(e)$ of effort $e$. The agent choose her effort level $e\in\{0,1\}$. For simplicity, let $c(1)=c>c(0)=0$. The utility function $u:\mathbb{R}^{+}\rightarrow \mathbb{R}$ is given by $u(W)=\sqrt{W}$.

            \indent The wage scheme $W: \Omega \rightarrow \mathbb{R}^{+}$ is designed based on the agent's performance data, which is modeled as a finite state space  $\Omega=\{\omega_{1},\omega_{2},\omega_{3}\}$. The probability of observing performance data $\omega_{1}$ is given by $\mu_{e}(\omega_{1})>0$, where $e\in\{0,1\}$. It is naturally to think of these probabilities as being affected by efforts. Note that $\sum_{\omega}\mu_{e}(\omega)=1$ for each effort level $e$.

            \indent I discuss the setting in which the agent's choice of effort cannot be observed by the principal (I refer to \cite{holmstrom1979moral} for more discussion of this setting). The principal keeps seeking the wage scheme that minimize the incentive cost plus the complexity cost. In other words, the principal faces the trade-off between designing a complex contract and saving the design cost. Moreover, suppose the complexity cost of $W$ is given by $\mathcal{C}(\sigma(W))=0$ for any $W$ such that $S^{W}=\{\Omega\}$, and $\mathcal{C}(\sigma(W))=\delta|S^{W}|$ otherwise. $|S^{W}|$ measures the number of elements in the partition induced by $W$, and $\delta>0$ measures the degree of complexity aversion.

            \indent Since the agent's effort level cannot be observed, the principal must make sure that the optimal wage scheme leads the agent to voluntarily choose desired effort level. Thus, if the principal wishes to induce the agent to exert high effort, the problem is
            \begin{equation*}
		    \underset{W}{\min}\sum_{\omega}\mu_{1}(\omega)W(\omega) + \delta|S^{W}| \quad \text{ subject to } 
            \end{equation*}
            \begin{equation}
		    \sum_{\omega}\mu_{1}(\omega)\sqrt{W(\omega)} - c\geq \bar{u}, \text{ and }
            \end{equation}
            \begin{equation}
		    \sum_{\omega}\mu_{1}(\omega)\sqrt{W(\omega)} - c\geq \sum_{\omega}\mu_{0}(\omega)\sqrt{W(\omega)}. 
            \end{equation}

            \indent The constraint (1) requires that the wage scheme yields the agent at least her reservation utility $\bar{u}$. The constraint (2) ensures that the effort level that the principal intends to induce is the same as that actually chosen by the agent. Let $\mu(\omega)=1- \frac{\mu_{0}(\omega)}{\mu_{1}(\omega)}$, constraint (2) can be rewritten as $\sum_{\omega}\mu_{1}(\omega)\mu(\omega)\sqrt{W(\omega)} \geq c$.

            \indent A complexity averse principal takes two steps to determine the optimal wage scheme. First, he chooses a contract that minimizes the cost for each complexity level. There are three types of contracts in this setting:
	        \begin{center}
	        \setlength{\tabcolsep}{6mm}{
		    \begin{tabular}{ll}
			Type & Contract \\
			\hline
			Simple contract&  $W(\omega)=W\ \forall \omega\in\Omega$\\
			Moderate complex contract & $W(\omega_{i})=W(\omega_{j})\neq W(\omega_{k})\ \forall i,j,k\in\{1,2,3\}$; \\
			Complex contract & $W(\omega_{1})\neq W(\omega_{2})\neq W(\omega_{3})$ \\
	        \end{tabular}}
            \end{center}
            \ \\
            Then, he determines the optimal contract which gives the minimal cost.

            \begin{proposition}   
            \hypertarget{prop1}{Suppose} for any $e\in\{0,1\}$, $\mu_{e}(\omega)\neq \mu_{e}(\omega^{\prime})$ for any $\omega\in\Omega$ and $\mu_{1}(\omega)\neq \mu_{0}(\omega)$ for at least one $\omega\in\Omega$. An optimal wage scheme that induces high effort\footnote{It is easy to see that a simple contract can only induce low effort.} exists and\\
            (1) If $\delta>\big[\frac{\mu_{1}(\omega_{k})(1-\mu_{1}(\omega_{k}))}{\mu^{2}(\omega_{k})}+1\big]c^{2}$, the principal chooses an optimal contract that is moderate complex. $W^{*}(\omega_{i})=W^{*}(\omega_{j})\neq W^{*}(\omega_{k})\ \forall i,j,k\in\{1,2,3\}$ and
            \begin{equation*}
	        \big(W^{*}(\omega_{i}),W^{*}(\omega_{k})\big)= \big((\bar{u} - c\mu_{1}(\omega_{k})/\mu(\omega_{k}))^{2}, (\bar{u} + c(1-\mu_{1}(\omega_{k}))/\mu(\omega_{k}))^{2}\big)
            \end{equation*}
            (2) If $\delta< \big[\frac{\mu_{1}(\omega_{k})(1-\mu_{1}(\omega_{k}))}{\mu^{2}(\omega_{k})}+1\big]c^{2}$, the principal chooses an optimal contract that is complex. $W(\omega_{1})\neq W(\omega_{2})\neq W(\omega_{3})$, 
            \begin{equation*}
	        W^{*}(\omega) = \bigg[\bar{u} - c\mu(\omega)\bigg(\sum_{\omega}\mu_{0}(\omega)\mu(\omega)\bigg)^{-1}\bigg]^{2}\quad  \forall \omega\in \Omega.
            \end{equation*}
            \end{proposition}

            \indent \hyperlink{prop1}{Proposition 1} shows that if the principal is complexity averse, he may not choose a complex wage scheme which gives each possible state of performance a different wage. In other words, if the principal is constrained in his ability to analyze all possible states of agent's performance, he tends to design and offer moderate complex wage scheme. Moreover, it is easy to see that as $\delta$ increases, the principal is more likely to choose a moderate complex wage scheme.

            \begin{proposition} 
            \hypertarget{prop2}{Suppose} a principal with the degree of complexity aversion $\delta$ chooses an optimal contract that is moderate complex, then so is one with the degree of complexity aversion $\delta^{\prime}$ if $\delta^{\prime}\geq \delta$.
            \end{proposition}

\subsection{Attention Reallocation}  \label{sec5.2}
\subsubsection{Introduction}\label{sec5.2.1}
            There are two types of axiomatic models of inattention: signal-based model\footnote{In the signal-based model, the agent has a prior that represents her initial beliefs. Each possible realization of the signal induces a corresponding posterior belief via Bayes' rule. The characterization of the preference requires that the agent chooses a signal that maximizes the expected utility net of attention cost. See \cite{caplin2015revealed}, \cite{de2017rationally}, \cite{matvejka2015rational} for more details.} and partition-based model. In the partition-based model, the agent has to deal with a two-stage decision problem. First, she chooses what information she tends to pay attention to (\cite{ellis2018foundations} calls this \textsl{subjective information}). Then, after the arrival of new information, she chooses the act that maximizes her expected utility conditional on the realized part of her subjective information. The two types of models have a key distinction that indicates different behavior implications. In the signal-based model, the agent is constrained to choose signals. Once a signal is chosen, the agent's decision is determined. In contrast, the agent of \cite{ellis2018foundations} chooses an information partition that determines her decision of acts. However, no research discusses the situation in which the agent reallocates her attention after the arrival of new information. 
            
            If we say that the act chosen conditional on the new information is the ``correct'' one, then because of limited attention, the behavior of choosing act conditional on the realized part of subjective information can be regarded as making mistakes. In this paper, we suggest that the agent is self-aware that herself has an attention constraint. Thus, when new information arrives, although she cannot improve her cognitive ability in a short period of time to use the new information to make a decision, she can update her subjective state space, which leads to a new set of subjective information that is more compatible with the arrival information. I discuss the theory of attention reallocation under the setting of \cite{ellis2018foundations}. I extend the existing studies in following ways. I propose a modified framework of choice under uncertainty to model inattentive decision maker where the agent reallocates her attention after she receives new information. She is unable to precisely observe the new information, however, she could use some identification strategies to keep or discard possibilities. Then, she can pay attention to a finer information partition with the decreasing of uncertainty. Therefore, she can make a decision closer to the "correct" one.
            
\subsubsection{Optimal Inattention}\label{sec5.2.2}
            In this section, I introduce the model of optimal inattention and discuss how can we incorporate the idea of minimal complexity updating into this model.

            The agent's choices are defined by \textsl{a conditional choice correspondence}. The agent chooses any acts in $c(B|E)$ when the new information is $E$. \cite{ellis2018foundations} characterizes axiomatically the following \hyperlink{eq3}{Optimal Inattention Representation} 
            \begin{equation}
            \hypertarget{eq3}{\hat{S}^{B} \in \arg \max _{S \in \mathbb{P}}\left[\sum_{E \in S} \mu_{\Omega}(E) \max _{f \in B} \int_{\omega \in \Omega} u(f(\omega))\mu_{E}(d\omega)-\mathcal{C}(\sigma(S)) \right]}
            \end{equation} 
            \begin{equation}
            \hypertarget{eq4}{c(B|E) = \arg \max _{f \in B} \int_{\omega \in \Omega} u(f(\omega))\mu_{\hat{S}^{B}(\succsim_E)}(d\omega)}
            \end{equation} 
            where choice problem $B\subset \mathcal{F}$ is a finite subset of acts, $\hat{S}^{B}\in \mathbb{P}$ is her \textsl{subjective information} when facing the decision problem $B$ given $\mu\in \Delta(\Omega)$. The agent faces a two-stage decision problem. At the first stage (\hyperlink{eq3}{Equation 3}), the agent chooses her \textsl{subjective information} $\hat{S}^{B}$. After the realization of event $E$, she subjectively chooses the set of possible states $\hat{S}^{B}(\succsim_E)$ based on her \textsl{subjective information}. Then she makes the decision conditional on $\hat{S}^{B}(\succsim_E)$.

            The way an inattentive agent processes new information is different from an agent who has full attention. Namely, when an event occurs, an agent who has no inattention problem could discard all other states that are not in the new event, and choose an act conditional on this event. However, an inattentive agent could not pay full attention to the new event, she might only ambiguously perceive some elements in her subjective information that contains the new event. An inattentive agent can be regarded as myopic, unable to observe the new event clearly. To illustrate the idea, consider the following investment example 3.
            \begin{example}
            \hypertarget{exp3}{Suppose} there are $8$ possible states of the economy. Assume that the agent chooses her \textsl{subjective information} $\hat{S}^{B}=\{s_1,s_2,s_3,s_4,s_5,\}$ at the first stage. If the agent is a standard expected utility maximizer, then after the arrival of new information $E=\{\omega_{3},\omega_{4},\omega_{5}\}$, she could make her decision conditional on the event $E$. If the agent is inattentive, she actually could not pay full attention to the event $E$. The best she can do is to discard $s_1$ and $s_5$. Therefore, the main difficulty here is to depict the agent's behaviors about choosing what is possible after the occurrence of the new event. 
            \end{example}

            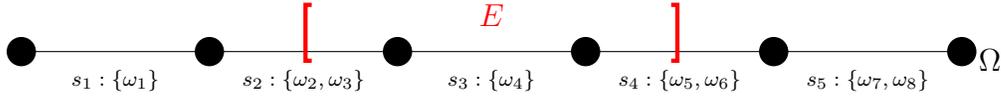
\begin{figure}[ht]
            \centering
            \begin{tikzpicture}[scale=1.25]
	        \tikzset{
		    mydot/.style={
			fill,
			circle,
		    }
	        }
	        \coordinate  (A) at(0,0) ;
	        \coordinate  (B)  at (2,0);
	        \coordinate  (C) at (4,0);
	        \coordinate  (D) at (6,0);
	        \coordinate  (E) at (8,0);
	        \coordinate  (F) at (10,0);
	         (3,0) coordinate (B) (6,0) coordinate (C) (9,0) coordinate (D);
	        \draw (A) -- node [midway, below=1mm]{\scriptsize $s_{1}:\{\omega_{1}\}$}
	        (B)  -- node [midway, below=1mm]{\scriptsize $s_{2}:\{\omega_{2},\omega_{3}\}$}
	        (C)   -- node [midway, below=1mm]{\scriptsize $s_{3}:\{\omega_{4}\}$}
	        (D)  -- node [midway, below=1mm]{\scriptsize $s_{4}:\{\omega_{5},\omega_{6}\}$}
	        (E)    -- node [midway, below=1mm]{\scriptsize $s_{5}:\{\omega_{7},\omega_{8}\}$}
	        (F);
	        \node[mydot,label={below left:}] at (A) {};
	        \node[mydot,label={below left:}] at (B) {};
	        \node[mydot,label={below left:}] at (C) {};
	        \node[mydot,label={below left:}] at (D) {};
	        \node[mydot,label={below left:}] at (E) {};
	        \node[mydot,label={below right:}] at (F) {};
	        \draw[{Bracket[width=8mm]}-,ultra thick,red]  ($(3,0)!-2mm!-90:(3.01,0)$)--node {} ($(3.01,0)!-2mm!90:(3,0)$);
	        \draw[-{Bracket[width=8mm]},ultra thick,red] ($(6.99,0)!-2mm!-90:(7,0)$)--node {} ($(7,0)!-2mm!90:(6.99,0)$);
	        \draw (5,0) node[above=2mm] {$\color{red}E$};
	        \draw (10.3,0) node[above=-4mm] {$\Omega$};
            \end{tikzpicture}
            \caption{Investment Example 3.}
            \label{fig:figure6}
            \end{figure}

            \indent Before  discussing the strategies that the agent will use to choose what is possible after the arrival of new information, I introduce the following notation. Recall that I denote by $\sigma(S)$ the algebra generated by partition $S\in \mathbb{P}$, that is, given $S\in \mathbb{P}$, $\sigma(S)$ is the family of events obtained by taking unions of elements in $S$.  

            \begin{definition}
            \hypertarget{def6}{For} any $B\in\mathcal{F}$ and $E\subseteq \Omega$, given $\hat{S}^{B}$, let $\hat{S}^{B}(\succsim_{E})\in \sigma(\hat{S}^{B})$ denote the set of states that the agent treats possible after the realization of $E$. A possibility selection rule is inattentive if   $\hat{S}^{B}(\succsim_{E})\in argmin\{|S^{B}(\succsim_{E})|: S^{B}(\succsim_{E})\in \sigma(\hat{S}^{B}), E\subseteq S^{B}(\succsim_{E})\}$.
            \end{definition}

            \indent In the structure of SEU, $\hat{S}^{B}(\succsim_{E})=E$ for any $E\in \Sigma$, the agent treats a state $\omega$ as impossible after the occurrence of event $E$ if $\omega\notin E$. However, the inattentive agent might treat subjectively a state $\omega$ as possible after the occurrence of event $E$, even if $\omega$ is not in $E$ objectively. For instance, $\hat{S}^{B}(\succsim_{E})$ can be $s_{1}\cup s_{2}\cup s_{3}\cup s_{4}$ or  $s_{2}\cup s_{3}\cup s_{4}\cup s_{5}$  in \hyperref[fig:figure6]{Figure 6} example. This is the key difference between an inattentive agent and SEU agent, which implies how an inattentive agent processes an event $E$.  

            \indent A noteworthy remark is that, here I do assume a possibility selection rules. The notation imposes a particular structure on the agent's behavior when choosing what is possible after the arrival of new information. It requires that the agent will discard all irrelevant cells of her subjective information. I argue that this does make sense, since I assume that the agent is self-aware that she has an attention constraint. 

            \noindent \textbf{Continue Example 3.} After the occurrence of event $E$, the agent treats $\hat{S}^{B}(\succsim_{E})=\{\omega_{2},\omega_{3},\omega_{4},\omega_{5},\omega_{6}\}$ as possible. Since states $\omega_{1}$ and $\omega_{7}$ are regarded as impossible, her new subjective state space is $\{\omega_{2},\omega_{3},\omega_{4},\omega_{5},\omega_{6}\}$. It is reasonable to infer that she will try to reallocate her attention on her new subjective state space. Consider another information partition $\bar{S}^{B}=\{\{\omega_{1}\},\{\omega_{2}\},\{\omega_{3}\},\{\omega_{4},\omega_{5},\omega_{6}\},\{\omega_{7}\},\{\omega_{8}\}\}$ which is finer than $\hat{S}^{B}$ on $\Omega$, but has the same level of complexity as  $\hat{S}^{B}$ on $E$. If she reallocate her attention to $\bar{S}^{B}$, $\bar{S}^{B}(\succsim_{E})=\{\omega_{3},\omega_{4},\omega_{5},\omega_{6}\}$, which means she can take another step and keep discarding $\omega_{2}$. In other words, it is possible that an inattentive agent could keep discarding state $\omega_{6}$ by reallocating her attention and get $S^{B}(\succsim_{E})=E$.

            In next section, I formalize the attention reallocation model and present main result. It shows that an inattentive agent can have a more accurate reasoning about the event $E$ by reallocating her attention to a new subjective information. 
            
\subsubsection{Attention Reallocation}\label{sec5.2.3}
            In the model of \hyperlink{eq5}{Minimal Attention Reallocation}, an inattentive agent faces a $T$-stage decision problem. 
            
            $\bullet$ At time $t=1$, she chooses her \textsl{subjective information} $S_{1}^{B}$ (the same as the behavior of agents in the model of optimal inattention, here I rewrite \hyperlink{eq3}{Equation 3} as \hyperlink{eq5}{Equation 5} but with the same meaning).
            \begin{equation}
            \hypertarget{eq5}{S^{B}_{1} \in \arg \max _{S \in \mathbb{P}}\left[\sum_{E \in S} \mu_{\Omega}(E) \max _{f \in B} \int_{\omega \in \Omega} u(f(\omega))\mu_{E}(d\omega)-\mathcal{C}(\sigma(S)) \right]}
            \end{equation} 
            \indent At time $T=2$, she receives the new information $E$ and selects $S_{1}^{B}(\succsim_{E})$ as the set of possible states. Then, she can choose a new subjective information to pay attention to. Before proceeding, we need to figure out two things. First, does she still choose subjective information in $\mathbb{P}$? Second, how does she measure the attention cost after the arrival of new information? To answer the first question, consider the following definition. $\mathbb{P}(\hat{S}^{B}(\succsim_{E}))$ requires that she will only consider information partitions that are compatible with what she considers possible after the realization of $E$.

            \begin{definition}
            \hypertarget{def7}{\textsl{For}} \textsl{any} $B\in\mathcal{F}$ \textsl{and} $E\subseteq \Omega$, \textsl{given} $\hat{S}^{B}$, \textsl{let} $\mathbb{P}(\hat{S}^{B}(\succsim_{E}))\in \mathbb{P}$ \textsl{denote the set of partitions that the agent might choose after the realization of} $E$. \textsl{An attention reallocation rule is} \textbf{adaptive} \textsl{if for all} $S^{B}\in \mathbb{P}(\hat{S}^{B}(\succsim_{E}))$, $\cup s = \hat{S}^{B}(\succsim_{E})$ \textsl{for some cells} $s$ in $S^{B}$.
            \end{definition}

            The \hyperlink{def8}{minimal conditional attention cost function} deals with the second question. Since all states outside of $S_{1}^{B}(\succsim_{E})$ are regarded as impossible, she does not have to consider the attention cost of those states for the new subjective information. Therefore, when she tries to measure her new subjective information $S$, she can choose an information partition $R$ that gives the same information structure as $S$ on $S_{1}^{B}(\succsim_{E})$ but with lowest attention costs on $\Omega$.
            \begin{definition}
            \hypertarget{def8}{\textsl{Given}} $E\in\Sigma^{\prime}$ \textsl{and} $\mu\in\Delta{\Omega}$, \textsl{for any} $B\in\mathcal{F}$, \textsl{a conditional attention cost function} $\mathcal{C}_{E,\mu}: \{\sigma(S): S\in\mathbb{P}\} \longrightarrow [0,+\infty)$ \textsl{is} \textbf{minimal} \textsl{if} $\mathcal{C}_{E,\mu}\big(\sigma(S)\big)=min\big\{\mathcal{C}\big(\sigma(R)\big)/\mu\big(S(\succsim_{E})\big): R\in \mathbb{P} , \textsl{ and }\sigma\big(R|S(\succsim_{E})\big)=\sigma\big(S|S(\succsim_{E})\big)\big\}$
            \end{definition}

            $\bullet$  At time $t=2$, she chooses $S_{2}^{B}$ as her stage$2$ subjective information (\hyperlink{eq6}{Equation 6}).
            \begin{equation}
            \hypertarget{eq6}{S^{B}_{2} \in \arg \max _{S \in \mathbb{P}(S_{1}^{B}(\succsim_{E}))}\left[\sum_{E \in S} \mu_{S^{B}_{1}(\succsim_{E})}(E) \max _{f \in B} \int_{ E \in \Omega} u(f(\omega))\mu_{E}(d\omega)-\mathcal{C}_{E,\mu}(\sigma(S)) \right]}
            \end{equation} 
            where $\mu_{S^{B}_{1}(\succsim_{E})}(E)=\frac{\mu(E \cap S^{B}_{1}(\succsim_{E}))}{\mu(S^{B}_{1}(\succsim_{E}))}$, and $\mathcal{C}_{E,\mu}\big(\sigma(S)\big)=min\big\{\mathcal{C}\big(\sigma(R)\big)/\mu\big(S^{1}_{B}(\succsim_{E})\big): R\in \mathbb{P} , \text{ and }\sigma\big(R|S^{1}_{B}(\succsim_{E})\big)=\sigma\big(S|S^{1}_{B}(\succsim_{E})\big)\big\}$.
            
            $\bullet$  Based on $S^{B}_{2}$, she selects $S_{2}^{B}(\succsim_{E})$ as the set of possible states for stage$2$. The agent will keep repeat this procedure until no states can be discarded. Suppose she will stop at stage$T$ (\hyperlink{eq7}{Equation 7}). 
            \begin{equation}
            \hypertarget{eq7}{S^{B}_{T} \in \arg \max _{S \in \mathbb{P}(S_{T-1}^{B}(\succsim_{E}))}\left[\sum_{E \in S} \mu_{S^{B}_{T-1}(\succsim_{E})}(E) \max _{f \in B} \int_{ E \in \Omega} u(f(\omega))\mu_{E}(d\omega)-\mathcal{C}_{E,\mu}(\sigma(S)) \right]}
            \end{equation} 
            
            $\bullet$  Finally, she makes the decision conditional on $S_{T}^{B}(\succsim_E)$ (\hyperlink{eq8}{Equation 8}).
            \begin{equation}
            \hypertarget{eq8}{c(B|E) = \arg \max _{f \in B} \int_{\omega \in \Omega} u(f(\omega))\mu_{S_{T}^{B}(\succsim_E)}(d\omega)}
            \end{equation} 

            \noindent \textbf{Revisiting Example 3.} I illustrate the \hyperlink{def7}{attention reallocation rule} and \hyperlink{def8}{minimal conditional attention cost function} by \hyperlink{exp3.2}{Example 3}.
            
            Recall that at stage$1$, the agent chooses $\hat{S}^{B}$. After the occurrence of event $E=\{\omega_{3},\omega_{4},\omega_{5}\}$, the agent selects $\hat{S}^{B}(\succsim_{E})=\{\omega_{2},\omega_{3},\omega_{4},\omega_{5},\omega_{6}\}$ as possible. Then she chooses to reallocate her attention to information partitions in $\mathbb{P}(\hat{S}^{B}(\succsim_{E}))$, e.g., $\{\{\omega_{1},\omega_{7},\omega_{8}\},\{\omega_{2}\},\{\omega_{3},\omega_{4},\omega_{5}\},\{\omega_{6}\}\}$ and $\{\{\omega_{1},\omega_{7}\},\{\omega_{2},\omega_{3}\},\{\omega_{4},\omega_{6}\},\{\omega_{5}\},\{\omega_{8}\}\}$. At stage$2$, suppose she reallocates her attention to information partition $\bar{S}^{B}=\{\{\omega_{1}\},\\\{\omega_{2}\},\{\omega_{3}\},\{\omega_{4},\omega_{5},\omega_{6}\},\{\omega_{7}\},\{\omega_{8}\}\}$. The attention cost of $\bar{S}^{B}$ is measured by information partition $R=\{\{\omega_{1},\omega_{4},\omega_{5},\omega_{6},\omega_{7},\omega_{8}\},\{\omega_{2}\},\{\omega_{3}\}\}$. By \hyperlink{def6}{Definition 6}, she selects $\bar{S}^{B}(\succsim_{E})=\{\omega_{3},\omega_{4},\omega_{5},\omega_{6}\}$ as possible. Compared to $\hat{S}^{B}(\succsim_{E})$, she achieves a more accurate understanding of event $E$ by reallocating her attention to a new subjective information. She repeats the process until no states can be discarded. 

            \begin{figure}[ht]
            \centering
	        \begin{tikzpicture}
		    \tikzset{
			mydot/.style={
				fill,
				circle,
				inner sep=1.5pt
			    }
		    }
		    \coordinate  (A) at(0,0) ;
	        \coordinate  (B)  at (1.5,0);
	        \coordinate  (C) at (3,0);
	        \coordinate  (D) at (4.5,0);
	        \coordinate  (E) at (6,0);
            \coordinate  (F)  at (7.5,0);
	        \coordinate  (G) at (9,0);
	        \coordinate  (H) at (10.5,0);
	        \coordinate  (I) at (12,0);
		     (3,0) coordinate (B) (6,0) coordinate (C) (9,0) coordinate (D);
		    \draw (A) -- node [midway, below=1mm]{\scriptsize $\omega_{1}$}
	        (B)  -- node [midway, below=1mm]{\scriptsize $\omega_{2}$}
            (C)  -- node [midway, below=1mm]{\scriptsize $\omega_{3}$}
	        (D)  -- node [midway, below=1mm]{\scriptsize $\omega_{4}$}
	        (E)  -- node [midway, below=1mm]{\scriptsize $\omega_{5}$}
            (F)  -- node [midway, below=1mm]{\scriptsize $\omega_{6}$}
	        (G)  -- node [midway, below=1mm]{\scriptsize $\omega_{7}$}
            (H)  -- node [midway, below=1mm]{\scriptsize $\omega_{8}$}
	        (I);
		    \node[mydot,label={below left:}] at (A) {};
	        \node[mydot,label={below left:}] at (B) {};
	        \node[mydot,label={below left:}] at (C) {};
	        \node[mydot,label={below left:}] at (D) {};
	        \node[mydot,label={below left:}] at (E) {};
	        \node[mydot,label={below left:}] at (F) {};
            \node[mydot,label={below left:}] at (G) {};
	        \node[mydot,label={below left:}] at (H) {};
	        \node[mydot,label={below right:}] at (I) {};
            \draw[{Bracket[width=28mm]}-,ultra thick,blue]  ($(1.5,0)!-2.5mm!-90:(1.51,0)$)--node {} ($(1.51,0)!-2.5mm!90:(1.5,0)$);
		    \draw[-{Bracket[width=28mm]},ultra thick,blue] ($(8.99,0)!-2.5mm!-90:(9,0)$)--node {} ($(9,0)!-2.5mm!90:(8.99,0)$);
		    \draw (5.25,0) node[above=11mm] {$\color{blue}\hat{S}^{B}(\succsim_{E})$};
		    \draw (12.5,0) node[above=-4mm] {$\Omega$};
		    \draw (12.5,0) node[above=9mm] {$\bar{S}^{B}_{1}$};
		    \draw (12.5,0) node[above=4mm] {\color{orange}$R$};
      		\draw[-,black,line width=2.8pt] ($(A)!-2mm!-90:(B)$)--node {} ($(B)!-2mm!90:(A)$);
		    \draw[-,black,line width=2.8pt] ($(B)!-4mm!-90:(C)$)--node {} ($(C)!-4mm!90:(B)$);
            \draw[-,black,line width=2.8pt] ($(C)!-6mm!-90:(D)$)--node {} ($(D)!-6mm!90:(C)$);
            \draw[-,black,line width=2.8pt] ($(D)!-8mm!-90:(G)$)--node {} ($(G)!-8mm!90:(D)$);
            \draw[-,black,line width=2.8pt] ($(G)!-10mm!-90:(I)$)--node {} ($(I)!-10mm!90:(G)$);
            \draw[-,orange,line width=3.8pt] ($(B)!-3.0mm!-90:(C)$)--node {} ($(C)!-3.0mm!90:(B)$);
		    \draw[-,orange,line width=3.8pt] ($(C)!-5.0mm!-90:(D)$)--node {} ($(D)!-5.0mm!90:(C)$);
		    \draw[-,orange,line width=3.8pt] ($(D)!-7mm!-90:(I)$)--node {} ($(I)!-7mm!90:(D)$);
		    \draw[-,orange,line width=3.8pt] ($(A)!-7mm!-90:(B)$)--node {} ($(B)!-7mm!90:(A)$);
	        \end{tikzpicture}
            \caption{Revisiting Investment Example 3.}
            \label{fig:figure7}
            \end{figure}
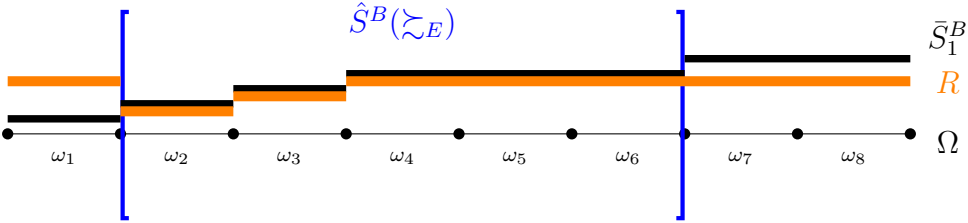
            
            one advantage of the model of attention reallocation is that the agent does not have to form a new cost function for partitions, the only thing she has to know is the attention cost on $\omega$
            
            The result shows that the agent can discard more irrelevant states by attention reallocation under some appropriate decision rules. \hyperlink{prop3}{Proposition 3} indicates that an inattentive agent is able to achieve a more sophisticated understanding of new information by attention reallocation. The key difference between the Attention Reallocation Representation and the Optimal Inattention representation is that subjective information is not just an information partition but  a ``belief" for the agent with Attention Reallocation Representation. She can use the new information to update her subjective information.  In the model of optimal inattention, the agent makes a decision based on her subjective information. There is actually no ``belief" updating, the agent does not try to have a more accurate understanding of the state space.  This is not reasonable, because if we all assume that inattention stems from redundant information, the agent must be self-aware of her inattention problem. Then why she does not try to do something to overcome the problem after the arrival of new information? My model provides an updating rule so that the agent can update her subjective information to try to overcome the inattention problem.

            \begin{proposition}\hypertarget{prop3}{Given} attention cost function, if $\romannumeral1)$ the \hyperlink{def6}{possibility selection rule} is inattentive; $\romannumeral2)$ the \hyperlink{def7}{attention reallocation rule} is adaptive; and $\romannumeral3)$ the \hyperlink{8}{conditional attention cost function} is minimal, then there exists some $B\in\mathcal{F}$, $E\in\Sigma^{\prime}$ and $\mu\in \Delta(\Omega)$, such that $S^{B}_{T}(\succsim_E)\subset S^{B}_{1}(\succsim_E)$.
            \end{proposition}
            
            \begin{proof} Suppose for some $B\in\mathcal{F}$, $E\in\Sigma^{\prime}$ and $\mu\in \Delta(\Omega)$, her stage$1$ subjective information is
            \begin{equation*}
            S^{B}_{1} = \{\{\omega_{1}\},\{\omega_{2},\omega_{3}\},\cdots,\{\omega_{l}\},\cdots\},\ l\in\{5,\cdots,n\},
            \end{equation*}
            where $n$ is the number of states. Let $E=\{\omega_{3},\cdots,\omega_{n-1}\}$, then
            her possibility set is
            \begin{equation*}
            S_{1}^{B}(\succsim_{E}) = \{\omega_{2},\cdots,\omega_{n-1}\}.
            \end{equation*}
            Consider an information partition
            \begin{equation*}
            S_{2}^{B} = \{\{\omega_{1}\},\{\omega_{2}\},\{\omega_{3},\cdots,\omega_{n-1}\},\{\omega_{n}\}\}.
            \end{equation*}
            The agent uses the attention cost of $R=\{\{\omega_{2}\},\{\omega_{1},\omega_{3},\cdots,\omega_{n}\}\}$ to measure the cost of $S_{2}^{B}$, the cost of which is very close to zero. That is
            \begin{equation*}
            \mathcal{C}_{E,\mu}(S_{1}^{B})\gg \mathcal{C}_{E,\mu}(S_{2}^{B}).
            \end{equation*}
            Therefore, the agent might choose $S_{2}^{B}$ at the stage$2$. So we have $S_{2}^{B}(\succsim_{E})\subset S_{1}^{B}(\succsim_{E})$. 
            \end{proof}

\section{Related Theoretical Literature} \label{sec6}

            There are several related strands of theoretical literature. The first are studies which provide axiomatic foundations for representations in which an agent has to face the trade-off between complex choices and saving on the cognitive resources of doing so. In \cite{puri2018preference}, an agent dislikes a lottery with more outcomes. \cite{ortoleva2013price} models the behavior of an agent who dislikes lotteries of menus of objects that have a larger number of menus because of the cost of thinking involved in choosing from them. The agent of \cite{ergin2010unique} considers the cost of contemplation before choosing an object from a menus. Instead of lotteries, I model preferences over acts. \cite{valenzuela2020subjective} introduces a model of preferences inspired by similar considerations to ours. The complexity of an act is measured by the cardinality of the partition induced by the act. In Valenzuela-Stookey’s Simple Bounds representation, not well-understood acts are mapped to the set of well-understood acts and then compared by their expected utilities. An act is called well-understood if the number of elements of it's partition equals to a cut-off. The key difference between \cite{valenzuela2020subjective} and this paper is that this paper does not restrict agents' understanding of acts. Instead, in my model, an agent evaluates an act using it's expected utility net of the complexity cost. To illustrate the difference of behavioral implications, consider two acts $f$ and $g$. $f$ is well understood and its' corresponding partition has two elements; $g$ is not well understood and its' corresponding partition is the state space. In \cite{valenzuela2020subjective}, if the simple greatest lower bound of $f$ is preferred to the simple least upper bound of $g$, then $f$ is preferred to $g$. However, in my model, the agent might prefer $g$ to $f$ because of the complexity cost of $f$.

            \indent Second are papers which focus on the interactive decision making situation. \cite{neyman1985bounded} studies the finitely repeated game in which only strategies that use a bounded number of states in the automaton available to players. \cite{ben1993repeated}, and \cite{megiddo1986play} follow this approach.  Instead of limiting the set of strategies, \cite{abreu1988structure} assume that more complex strategies means higher costs.  \cite{rubinstein1986finite}, and \cite{abreu1988structure} restrict players’ strategies in a repeated game to those implementable by finite state automata. \cite{mengel2012learning} studies the learning process of two players who face many games. Since it requires too much cognitive resources to distinguish all games, players choose to partition the set of all games into categories. My setting is closest to \cite{abreu1988structure}.

            \indent I use the corresponding partition of an act to measure complexity which is interpreted as coarse understanding of the state space. However, the agent in my model fully understands the state space. Here, the coarseness is induced by her constraints in cognitive ability to consider the complexity of acts. I discuss two strands of related studies on coarse contingencies and show how these models differ from my model.

            \indent The first are studies which focus on coarse understanding and ambiguity. The works most relevant to my model are \cite{ahn2010framing} and \cite{epstein2007coarse}.\footnote{For more details about subjective contingencies and ambiguity, see \cite{dekel2001representing}; \cite{dillenberger2014theory}; \cite{ghirardato2001coping}; \cite{minardi2019subjective}; \cite{mukerji1997understanding}; and \cite{saponara2022revealed}.}

            \indent \cite{ahn2010framing} proposes a model of decision  making under uncertainty in which the primitive is a class of preference relations indexed by partitions of the state space. If two acts respect the same partition, then the agent ranks the two acts by their expected utilities based on her partition-dependent belief. The agent in my model behaves similarly but with respect to partition-independent belief (prior). How does the agent in the model of \cite{ahn2010framing} ranks two acts that respect different partitions, e.g., two acts $f$ and $g$ in the investment example? They define that the agent compares $f$ and $g$ by computing their expected utilities based on the coarsest common refinement of $S^{f}$ and $S^{g}$. In the investment example, the coarsest common refinement of $S^{f}$ and $S^{g}$ is $S^{f}$. If the corresponding  partition-dependent belief is $(\frac{1}{3},\frac{1}{3},\frac{1}{3})$, she prefers $f$ to $g$. However, in my model, a complexity averse agent prefers $g$ to $f$. From above example, we can conclude two main different aspects between my model and theirs. First, a complexity averse agent's belief has nothing to do with partitions. Second, a complexity averse agent compares two acts by computing expected utility net of complexity cost. She does not try to identify two acts in the same partition. 
 
            \indent \cite{epstein2007coarse} models an agent who forms some contingencies and is self-aware of the coarseness of these contingencies. They show that coarse contingencies induce a preference for hedging, as in the ambiguity aversion studies.\footnote{For more studies about ambiguity aversion, see \cite{gilboa1989maxmin}; \cite{schmeidler1989subjective}.} The key difference between my model and \cite{epstein2007coarse} is that in the latter paper the coarse contingencies are exogenous. In this model, the coarseness is endogenous and is induced by her aversion to complexity. Another noteworthy remark is that, complexity aversion does not induce a preference for hedging. In my setting, the mixture of two acts might increase or decrease the complexity level, thus, the agent in this model does not exhibit uncertainty aversion (the key axiom in the model of \cite{gilboa1989maxmin}).
 
            \indent Second are papers about rational inattention which is introduced by \citep{sims1998stickiness, sims2003implications}. In this model, attention cost is interpreted as the expected difference between the prior uncertainty about the state and the posterior uncertainty. And this theory has been applied to many economic problems.\footnote{For instance, there are many studies that apply to consumption-savings problems: \cite{sims2006rational}; \cite{mackowiak2015business}. There are also many studies that apply the model to the theory of price setting: \cite{luo2008consumption}; \cite{mackowiak2009optimal}.} However, I am more interested in axiomatic models of inattention. \cite{de2017rationally} provide an axiomatic characterization of rationally inattentive preferences over menus. \cite{ellis2018foundations} introduces a representation of preferences similar to theirs but takes a choice correspondence as a primitive. Other related studies are \cite{dillenberger2014theory};  \cite{lu2016random}. However, no studies discuss the situation in which the agent reallocates her attention after the arrival of new information. This paper fills this gap.  

            \indent There are two key differences between complexity aversion and rational inattention. First, in my model, the corresponding partition of an act is only used to measure the complexity cost of this act. A complexity averse agent has full attention and directly chooses acts both at ex-ante and ex-post stage. Instead, an optimal inattentive agent chooses what to pay attention to at the ex-ante stage, she is unable to make the decision before the arrival of information. Second, the corresponding partition of an act, chosen by a complexity averse agent, is not the same information partition that an optimal inattentive agent will choose. To see this, suppose that a complexity averse agent chooses an act $f$ at ex-ante stage such that the corresponding partition is $\{\{\omega_{1}\},\{\omega_{2},\omega_{3},\omega_{4}\},\{\Omega_{5}\}\}$. But, an inattentive agent might have a more accurate understanding of $f$, since  $\{\{\omega_{1}\},\{\omega_{2},\omega_{3}\},\{\omega_{4}\},\{\Omega_{5}\}\}$ gives the same expected utility with $f$. In other words, a complexity averse agent does not try to find an optimal partition. 

\newpage
    \bibliography{mybib}
	\bibliographystyle{apalike}   

\newpage
\section*{Appendix} \label{app}
\begin{appendices}
\section{Proofs of Section 3.} \label{appA}
\subsection{Proof of Theorem 1.} \label{appA.1}
            
            For the necessity part, it is straightforward to show that a complexity averse preference $\succsim$ satisfies \hyperlink{wo}{Axiom 1-6}. The proof of \textsl{Only if} part proceeds as a sequence of Lemmas.

            \begin{lemma} 
            \hypertarget{lem1}{For} any $f\in \mathcal{F}$, if there exists $x\in \overline{\mathcal{F}}$ with $f\succsim x$, then there exist a $x_{f}\in \overline{\mathcal{F}}$ such that $x_{f}\sim f$. 
            \end{lemma}
            
            \begin{proof} 
            Let $x^{\prime}$ be a best outcome and $x^{\prime\prime}$ be a worst outcome that some acts induce. By \hyperlink{wm}{Axiom 4}, we have $x^{\prime}\succsim f\succsim x^{\prime\prime}$. Then, by \hyperlink{continuity}{Axiom 2}, the following two sets are closed
            \begin{equation*}
	        A_{1}=\{\alpha \in[0,1]: \alpha x^{\prime}+(1-\alpha) x^{\prime\prime} \succsim f\} \quad \text { and } \quad A_{2}=\{\alpha \in[0,1]: f \succsim \alpha x^{\prime}+(1-\alpha) x^{\prime\prime}\}
            \end{equation*}
            Since $A_{1}\cup A_{2}$ is connected, we must have $\alpha\in A_{1}\cap A_{2}$ such that $x_{f}=\alpha x^{\prime}+(1-\alpha) x^{\prime\prime} \sim f$.
            \end{proof}
            
            \newpage
            \begin{lemma} 
            \hypertarget{lem2}{There} exist an affine utility function $u: X\rightarrow \mathbb{R}$ with unbounded range and a prior probability measure $\mu$ over $\Omega$ such that
            \begin{equation*}
            U(f)=\int_{\Omega} u (f(\omega)) \mu(d\omega).
            \end{equation*} 
            \end{lemma}
            
            \begin{proof} 
            For $\succsim$ on $\mathcal{F}$, by \hyperlink{wci}{Axiom 3}, if $x\sim y$, we have $\alpha x+(1-\alpha) y \sim \alpha y+(1-\alpha) y$. Then by \hyperlink{ac}{Axiom 5}, $\alpha x+(1-\alpha) z \succsim \alpha y+(1-\alpha) z$. Using \hyperlink{wci}{Axiom 3} again, we have $\alpha y+(1-\alpha) y \sim \alpha x+(1-\alpha) y$. By \hyperlink{ac}{Axiom 5}, $\alpha y+(1-\alpha) z \succsim \alpha x+(1-\alpha) z$. So we can conclude that if $x\sim y$, we have $\alpha x+(1-\alpha) z \sim \alpha y+(1-\alpha) z$. 
            
            \indent  Moreover, by \hyperlink{lem1}{Lemma 1}, we can find $x_{f}\sim f$ for all $f\in \{f\in \mathcal{F}:\exists x\in \overline{\mathcal{F}} \text{ s.t } f\succsim x\}$. Then by \hyperlink{unb}{Axiom 6}, we have the unboundedness of affine utility function. Then let 
            \begin{equation*}
 	        U(f)=\left\{\begin{array}{cl}\int_{\Omega} u (x_{f}(\omega)) \mu(d\omega) & \text { if } f\in \{f\in \mathcal{F}:\exists x\in \overline{\mathcal{F}} \text{ s.t } f\succsim x\}, \\ -\infty & \text { if } \text{otherwise}.\end{array}\right.
            \end{equation*}
            
            \noindent Just consider $g\in \{f\in \mathcal{F}:\exists x\in \overline{\mathcal{F}} \text{ s.t } f\succsim x\}$. $f\succsim g\Leftrightarrow x_{f}\succsim x_{g}$. By \hyperlink{lem1}{Lemma 1}, we have the result. Thus, together with \hyperlink{wo}{Axiom 1} and \hyperlink{continuity}{Axiom 2}, by Theorem 8 (\cite{herstein1953axiomatic}), a measurable utility can be defined on $\mathcal{F}$. And combining Monotonicity, the DM is a standard expected utility maximizer (\cite{anscombe1963definition}). 
            \end{proof}
            \ \\
            \begin{lemma} 
            \hypertarget{lem3}{If} $f\in \{f\in \mathcal{F}:\exists x\in \overline{\mathcal{F}} \text{ s.t } f\succsim x\}$, and $\sigma(f)=\sigma(g)$, then $g\in \{f\in \mathcal{F}:\exists x\in \overline{\mathcal{F}} \text{ s.t } f\succsim x\}$. 
            \end{lemma}

            \begin{proof}
            Construct $f^{\prime}\in \mathcal{F}$ and $x^{\prime}\in \overline{\mathcal{F}}$ such that $u(f^{\prime}(\omega))=2u(f(\omega))$ and $u(x^{\prime}(\omega))=0$ for all $\omega\in \Omega$. Then, we have $\frac{1}{2}u(f^{\prime}(\omega))+\frac{1}{2}u(x^{\prime}(\omega))=u(f(\omega))$, by \hyperlink{lem2}{Lemma 2}, it can be rewrite as $u(\frac{1}{2}f^{\prime}(\omega)+\frac{1}{2}x(\omega))=u(f(\omega))$. Again, by \hyperlink{lem2}{Lemma 2}, we have $\frac{1}{2}f^{\prime}+\frac{1}{2}x^{\prime}\sim f$. Using order and continuity, we can find $x^{\prime\prime}\in \overline{\mathcal{F}}$ such that $\frac{1}{2}x^{\prime\prime}+\frac{1}{2}x^{\prime}\sim \frac{1}{2}f^{\prime}+\frac{1}{2}x^{\prime}$. Then by \hyperlink{wci}{Axiom 3}, for any $f\in \mathcal{F}$, we have $\frac{1}{2}x^{\prime\prime}+\frac{1}{2}f\sim \frac{1}{2}f^{\prime}+\frac{1}{2}f$. Therefore, we can choose $f\in \mathcal{F}$, such that $u(g(\omega))=u(\frac{1}{2}f^{\prime}(\omega)+\frac{1}{2}f(\omega))$. By \hyperlink{lem2}{Lemma 2}, $g\sim \frac{1}{2}f^{\prime}+\frac{1}{2}f\sim \frac{1}{2}x^{\prime\prime}+\frac{1}{2}f$. By Compete and \hyperlink{lem1}{Lemma 1}, it is easy to find $g\in \mathcal{F}$ and $z\in\overline{\mathcal{F}}$, such that $\frac{1}{2}x^{\prime\prime}+\frac{1}{2}f\sim g\succsim \frac{1}{2}x^{\prime\prime}+\frac{1}{2}z$, where $f\succsim z$ with some $z\in\overline{\mathcal{F}}$. And by order and continuity, there exists $y\in \overline{\mathcal{F}}$, such that $\frac{1}{2}x^{\prime\prime}+\frac{1}{2}z\sim y$. Finally, we find such $y\in \overline{\mathcal{F}}$, such that $g\succeq y$. Thus, $g\in \{f\in \mathcal{F}:\exists x\in \overline{\mathcal{F}} \text{ s.t } f\succsim x\}$.
            \end{proof}
            \ \\
            \begin{lemma} 
            \hypertarget{lem4}{Given} any $S\in\mathbb{P}$, there exists $f\in \{f\in \mathcal{F}:\exists x\in \overline{\mathcal{F}} \text{ s.t } f\succsim x\}$ such that $\sigma(f) = \sigma(S)$.
            \end{lemma}

            \begin{proof} 
            By contradiction, suppose there exists such $S\in\mathbb{P}$ such that all $f$ that satisfies $\sigma(f) = \sigma(S)$ is not in $\{f\in \mathcal{F}:\exists x\in \overline{\mathcal{F}} \text{ s.t } f\succsim x\}$. That is 
            \begin{equation*}
            \sigma(f) \neq \sigma(S) \text{ for all } f\in \{f\in \mathcal{F}:\exists x\in \overline{\mathcal{F}} \text{ s.t } f\succsim x\}.
            \end{equation*}
            Consider another $S^{\prime}\in\mathbb{P}$, there exists $f\in \{f\in \mathcal{F}:\exists x\in \overline{\mathcal{F}} \text{ s.t } f\succsim x\}$ such that $\sigma(f) = \sigma(S^{\prime})$. That is we can find an act $g$ such that $g\succsim x$.
            Therefore, it is easy to construct an act $f^{\prime}$ that is $f\sim g$ but with corresponding partition $S$.
            \end{proof}
\ \\            
            \begin{lemma} 
            \hypertarget{lem5}{There} exists a cost function $\mathcal{C}: \{\sigma(S):S\in \mathbb{P}\}\rightarrow \overline{\mathbb{R}}_{+}$ such that
            \begin{equation*}
            V(f)=\int_{\Omega} u (f(\omega)) \mu(d\omega)-\mathcal{C}(\sigma(f)) 
            \end{equation*}
            \noindent where $\sigma(S)=\cap\{\mathcal{A}\subset \Sigma: S\subset \mathcal{A} \text{ and } \mathcal{A} \text{ is a } \sigma \text{-algbra}\}$, $\mathcal{A}$ is a $\sigma$-algebra of subsets of $\Omega$; and $\overline{\mathbb{R}}_{+} \equiv[0, \infty]$.
            \end{lemma}

            \begin{proof} 
            Consider a information partition $S\in \mathbb{P}$, by \hyperlink{lem4}{Lemma 4}, there exists $f\in \{f\in \mathcal{F}:\exists x\in \overline{\mathcal{F}} \text{ s.t } f\succsim x\}$ with $\sigma(S)=\sigma(f)$. Construct an act $g_{S}\in \mathcal{F}$ such that $\sigma(g_{S})=\sigma(S)$. For example, suppose $S=\{s_{1},\cdots,s_{n}\}$:
            \begin{equation*}
	        g_{S}(\omega)=\left\{\begin{array}{cc}x_{1} & \text { if } \omega \in s_{1}, \\ \vdots\\x_{i}& \text { if } \omega \in s_{i}, \\ \vdots\\x_{n} & \text { if } \omega \in s_{n}.\end{array}\right.
            \end{equation*}
            where $x_{i}\neq x_{j}$ for all $i,j\in\{1,\cdots,n\}$. So we have $\sigma(g_{ S})=\sigma(f)$. By \hyperlink{lem3}{Lemma 3}, $g_{ S}\in\{f\in \mathcal{F}:\exists x\in \overline{\mathcal{F}} \text{ s.t } f\succsim x\}$. Then, by \hyperlink{lem1}{Lemma 1}, we can find $x_{S}\in \overline{\mathcal{F}}$ be such that $\frac{1}{2}x_{S}+\frac{1}{2}x^{\prime}\sim \frac{1}{2}g_{S}+\frac{1}{2}x^{\prime}$ ($x^{\prime}$ is the same as what we defined in the proof of \hyperlink{lem3}{Lemma 3}.). Then, we can pick $f^{\prime}\in \mathcal{F}$ be such that $u(f^{\prime}(\omega))=2u(f(\omega))$ for any $f\in \{f\in \mathcal{F}:\exists x\in \overline{\mathcal{F}} \text{ s.t } f\succsim x\}$ and any $\omega\in \Omega$. According to this construction, we have $\sigma(f)=\sigma(f^{\prime})$. Then by \hyperlink{lem2}{Lemma 2}, we have
            \begin{equation*}
            f\sim \frac{1}{2}g_{S}+\frac{1}{2}f^{\prime}\sim \frac{1}{2}x_{S}+\frac{1}{2}f^{\prime}.
            \end{equation*}
            \indent Thus, we have $U(f)=U(\frac{1}{2}x_{S}+\frac{1}{2}f^{\prime})=\frac{1}{2}U(x_{S})+\frac{1}{2}U(f^{\prime})=\frac{1}{2}U(x_{S})+\frac{1}{2}U(f^{\prime})$, that is
            \begin{equation*}
            U(f)=\frac{1}{2}U(x_{S})+\frac{1}{2}\int_{\Omega} u (2f(\omega)(d\omega))=\frac{1}{2}u(x_{S})+\int_{\Omega} u (f(\omega)(d\omega)).
            \end{equation*}
            We can define $\mathcal{C}(\sigma(S))=-\frac{1}{2}U(x_{S})$, then $U(P)=V(P)$.
            \end{proof}
            
            \newpage
            \begin{lemma} 
            \hypertarget{lem6}{If} $\succsim$ satisfies \hyperlink{ac}{Axiom 5}, then $\sigma(S)\subset \sigma(S^{\prime})$ implies $\mathcal{C}(\sigma(S))\leq \mathcal{C}(\sigma(S^{\prime})$.
            \end{lemma}

            \begin{proof} 
            Consider $\alpha x + (1-\alpha) f \sim \beta x + (1-\beta) f$, by \hyperlink{lem2}{Lemma 2}, we have
            \begin{equation*}
	        \begin{split}
	        &\int_{\Omega} u (x(\omega)) \mu(d\omega)=	\int_{\Omega} u (f(\omega)) \mu(d\omega)\\
	        & \Longrightarrow \lambda\int_{\Omega} u (x(\omega)) \mu(d\omega)=	\lambda\int_{\Omega} u (f(\omega)) \mu(d\omega)\\
	        & \Longrightarrow \lambda\int_{\Omega} u (x(\omega)) \mu(d\omega) + (1-\lambda)\int_{\Omega} u (g(\omega)) \mu(d\omega) \\
	        & =	\lambda\int_{\Omega} u (f(\omega)) \mu(d\omega) + (1-\lambda)\int_{\Omega} u (g(\omega)) \mu(d\omega) \\
	        \end{split}
            \end{equation*}
            Since we have $\lambda x + (1-\lambda)g\succsim \lambda f + (1-\lambda) g$, that is $\sigma(g)\subset \sigma(\lambda f + (1-\lambda) g)$ implies
            \begin{equation*}
	        \mathcal{C}(\sigma(g))\leq \mathcal{C}(\sigma(\lambda f + (1-\lambda) g))
            \end{equation*}
            by the fact that $\mathcal{C}(\sigma(g)) = \mathcal{C}(\sigma(\lambda x + (1-\lambda)g)$.
            \end{proof}

            \indent This completes the proof of sufficiency.
            
\subsection{Proof of Corollary 1.} \label{appA.2}

            \begin{proof}
            Suppose $\succsim$ is a complexity aversion preference represented by $\langle u,\mu,\mathcal{C}\rangle$. Then by the proof of \hyperlink{thm1}{Theorem 1}, $\int_{\Omega} u (f(\omega)(d\omega)) - \mathcal{C}\big(\sigma(f)\big)$ represents $\succsim$. Suppose for contradiction that there exists another complexity cost function $\mathcal{C}^{\prime}$ where $\int_{\Omega} u (f(\omega)(d\omega)) - \mathcal{C}^{\prime}\big(\sigma(f)\big)$ represents $\succsim$. Suppose $\mathcal{C}^{\prime}\big(\sigma(f)\big)>\mathcal{C}\big(\sigma(f))\big)$ for all $f\in\mathcal{F}^{S^{f}}$ and $\mathcal{C}^{\prime}\big(\sigma(g)\big)=\mathcal{C}\big(\sigma(g)\big)$ for all other acts $g\notin\mathcal{F}^{S^{f}}$.

            \indent For an act $f^{*}\in\mathcal{F}^{S^{f}}$, we construct a new act $f^{\prime}\in\mathcal{F}$ such that $f^{\prime}(\omega) = f^{*}(\omega) - \epsilon$ for all $\omega\in\Omega$ and a very small but positive $u(\epsilon)$. Suppose $f^{\prime}\sim g$ where $g\in\mathcal{F}$ but $g\notin\mathcal{F}^{S^{f}}$, by $\sigma(f^{*})=\sigma(f^{\prime})$, we know $f^{*}\succ g$. That is
            \begin{equation*}
	        \int_{\Omega} u (f^{*}(\omega)(d\omega)) - \mathcal{C}\big(\sigma(f^{*})\big)> \int_{\Omega} u (g(\omega)(d\omega)) - \mathcal{C}\big(\sigma(g)\big).
            \end{equation*}
            By $\mathcal{C}^{\prime}\big(\sigma(g)\big)=\mathcal{C}\big(\sigma(g)\big)$, we have 
            \begin{equation*}
	        \int_{\Omega} u (f^{*}(\omega)(d\omega)) - \mathcal{C}\big(\sigma(f^{*})\big)> \int_{\Omega} u (g(\omega)(d\omega)) - \mathcal{C}^{\prime}\big(\sigma(g)\big).
            \end{equation*}
            If $\mathcal{C}^{\prime}\big(\sigma(f)\big)>>\mathcal{C}\big(\sigma(f^{*})\big)$, we may have 
            \begin{equation*}
	        \int_{\Omega} u (f^{*}(\omega)(d\omega)) - \mathcal{C}^{\prime}\big(\sigma(f^{*})\big)< \int_{\Omega} u (g(\omega)(d\omega)) - \mathcal{C}^{\prime}\big(\sigma(g)\big).
            \end{equation*}
            Until now, we still can find a $\mathcal{C}^{\prime}$ that represents the same preference $f^{*}\succ g$. However, consider following
            \begin{equation*}
	        \begin{split}
	        &\int_{\Omega} u (f^{\prime}(\omega)(d\omega)) - \mathcal{C}\big(\sigma(f^{\prime})\big)= \int_{\Omega} u (g(\omega)(d\omega)) - \mathcal{C}\big(\sigma(g)\big).\\
	        &\Longleftrightarrow \int_{\Omega} u (f^{*}(\omega)(d\omega)) - u(\epsilon) - \mathcal{C}\big(\sigma(f^{*})\big) = \int_{\Omega} u (g(\omega)(d\omega)) - \mathcal{C}\big(\sigma(g)\big).\\
	        &\Longleftrightarrow \mathcal{C}\big(\sigma(f^{*})\big) + u(\epsilon) - \mathcal{C}\big(\sigma(g)\big) > \mathcal{C}^{\prime}\big(\sigma(f^{*})\big) - \mathcal{C}^{\prime}\big(\sigma(g)\big)\\
	        &\Longleftrightarrow  \mathcal{C}\big(\sigma(f^{*})\big) + u(\epsilon) > \mathcal{C}^{\prime}\big(\sigma(f^{*})\big)\\
	        \end{split}
            \end{equation*}
            This holds for any very small $u(\epsilon)$, so we have $\mathcal{C}\big(\sigma(f^{*})\big)=  \mathcal{C}^{\prime}\big(\sigma(f^{*})\big)$, a contradiction. The same logic to prove for a complexity cost function $\mathcal{C}^{\prime}$ that gives more different costs of acts compared to $\mathcal{C}$. 
            \end{proof}
            
\subsection{Proof of Corollary 2.}  \label{appA.3}

            \begin{proof}
            Suppose $\langle u,\mu, \mathcal{C}\rangle$ and $\langle u^{\prime}, \mu^{\prime}, \mathcal{C}^{\prime}\rangle$ represent the same preferences relations, and $\mathcal{C}$ and $\mathcal{C}^{\prime}$ are canonical.  By \hyperlink{lem2}{Lemma 2}, the preference relation has an expected utility representation, so $\mu=\mu^{\prime}$ and $\exists \beta_{1}>0$ and $\beta_{2}\in \mathbb{R}$ such that $u=\beta_{1} u^{\prime} + \beta_{2}$.\footnote{For the proof of the uniqueness of $\mu$ and $u$, we refer to Fishburn(1970).}  

            \indent Then we turn to prove $\mathcal{C} = \alpha \mathcal{C}$. Suppose $u$ and $u^{\prime}$ represent the same preference relations. Consider act $f\in\mathcal{F}$ such that there exists $x\in\overline{\mathcal{F}}$ with $f\succsim x$. By \hyperlink{lem1}{Lemma 1} there exist a $x_{f}\in\overline{\mathcal{F}}$ such that $x_{f}\sim f$. So we have 
            \begin{equation*}
            \int_{\Omega} u (x) \mu(d\omega) = \int_{\Omega} u \big( f(\omega)\big) \mu(d\omega)- \mathcal{C}\big(\sigma(f)\big)
            \end{equation*}
            and 
            \begin{equation*}
		    \begin{split}
	        & \int_{\Omega} u^{\prime} (x) \mu(d\omega) = \int_{\Omega} u^{\prime} \big( f(\omega)\big) \mu(d\omega)- \mathcal{C}^{\prime}\big(\sigma(f)\big)\\
	        & \Longrightarrow \alpha\int_{\Omega} u (x) \mu(d\omega) + \beta = \alpha\int_{\Omega} u \big( f(\omega)\big) \mu(d\omega) + \beta- \mathcal{C}^{\prime}\big(\sigma(f)\big)\\
	        & \Longrightarrow \alpha\big(\int_{\Omega} u (x) \mu(d\omega) - \int_{\Omega} u \big( f(\omega)\big) \mu(d\omega) \big) = \mathcal{C}^{\prime}\big(\sigma(f)\big)\\
	        & \Longrightarrow \alpha\mathcal{C}^{\prime}\big(\sigma(f)\big) = \mathcal{C}^{\prime}\big(\sigma(f)\big). 
		    \end{split}
            \end{equation*}
            \end{proof}

\subsection{Proofs of Theorem 2 } \label{appA.4}

            \begin{proof}
            \textsl{Only if part}. Suppose $\succsim^{1}$ has higher degree of complexity aversion than $\succsim^{2}$. If $x\succ^{1} f$, then  
            \begin{equation*}
	        \int_{\Omega} u^{1} (x(\omega)) \mu^{1}(d\omega) - \mathcal{C}^{1}\big(\sigma(x)\big)>  \int_{\Omega} u^{1} (f)(\omega)) \mu^{1}(d\omega)  - \mathcal{C}^{1}\big(\sigma(f)\big)
            \end{equation*}
            Since $x\succ^{1} f$ implies $x\succ^{2} f$, then
            \begin{equation*}
	        \int_{\Omega} u^{2} (x(\omega)) \mu^{2}(d\omega) - \mathcal{C}^{2}\big(\sigma(x)\big)>  \int_{\Omega} u^{2} (f(\omega)) \mu^{2}(d\omega)  - \mathcal{C}^{2}\big(\sigma(f)\big)
            \end{equation*}
            We know $(u^{1}, \mu^{1}) = (u^{2}, \mu^{2})$ and $\mathcal{C}\big(\sigma(x)\big)=0$ for all $x\in \overline{\mathcal{F}}$, then
            \begin{equation*}
	        \int_{\Omega} u^{1} (x(\omega)) \mu^{1}(d\omega) >  \int_{\Omega} u^{1} (f(\omega)) \mu^{1}(d\omega)  - \mathcal{C}^{1}\big(\sigma(f)\big) \geq  \int_{\Omega} u^{2} (f(\omega)) \mu^{1}(d\omega)  - \mathcal{C}^{2}\big(\sigma(f)\big)
            \end{equation*}
            Thus, we have $\mathcal{C}^{1}\big(\sigma(f)\big)\leq \mathcal{C}^{2}\big(\sigma(f)\big)$. 
            
            \noindent\textsl{If part}. If $\mathcal{C}^{1}\big(\sigma(f)\big)\leq \mathcal{C}^{2}\big(\sigma(f)\big)$, then
            \begin{equation*}
	        \int_{\Omega} u^{1} (x(\omega)) \mu^{1}(d\omega) >  \int_{\Omega} u^{1} (f(\omega)) \mu^{1}(d\omega)  - \mathcal{C}^{1}\big(\sigma(f)\big) \geq  \int_{\Omega} u^{2} (f(\omega)) \mu^{1}(d\omega)  - \mathcal{C}^{2}\big(\sigma(f)\big)
            \end{equation*}
            \end{proof}

\subsection{Proof of Theorem 3 } \label{appA.5}

            \begin{proof}
            \textsl{Only if part}. Suppose $\succsim^{1}$ has higher capacity for more complex acts than $\succsim^{2}$. If $supp(\mathcal{C}^{1})\nsubseteq supp(\mathcal{C}^{2})$, then we can find an act $h\in supp(\mathcal{C}^{2})\backslash supp(\mathcal{C}^{1})$. Since $h\notin supp(\mathcal{C}^{1})$, we cannot find an act $h^{\prime}\in supp(\mathcal{C}^{1})$ such that $\sigma(h)\subset \sigma(h^{\prime}) $, which implies $\mathcal{C}^{1}\big(\sigma(h^{\prime})\big)\geq \mathcal{C}^{2}\big(\sigma(h)\big)$.

            \indent Since $\alpha f + (1-\alpha) x \succ^{1}  \alpha f + (1-\alpha) g$ implies $\alpha f + (1-\alpha) x \succ^{2}  \alpha f + (1-\alpha) g$, we have
            \begin{equation*}
	        \mathcal{C}^{1}\Big(\sigma\big(\alpha f + (1-\alpha) x\big)\Big) - \mathcal{C}^{1}\Big(\sigma\big(\alpha f + (1-\alpha) g\big)\Big) \geq  \mathcal{C}^{2}\Big(\sigma\big(\alpha f + (1-\alpha) x\big)\Big) - \mathcal{C}^{2}\Big(\sigma\big(\alpha f + (1-\alpha) g\big)\Big).
            \end{equation*}

            Adding $\mathcal{C}^{1}\Big(\sigma\big(\alpha f + (1-\alpha) g\big)\Big)$ to both sides, we get
            \begin{equation*}
	        \begin{split}
		    & \mathcal{C}^{1}\Big(\sigma\big(\alpha f + (1-\alpha) g\big)\Big) + 
		    \mathcal{C}^{1}\Big(\sigma\big(\alpha f + (1-\alpha) x\big)\Big) -
		    \mathcal{C}^{1}\Big(\sigma\big(\alpha f + (1-\alpha) g\big)\Big) \\
		    & \ \ \ \ \ \ \geq   \mathcal{C}^{2}\Big(\sigma\big(\alpha f + (1-\alpha) x\big)\Big) + 
		    \mathcal{C}^{1}\Big(\sigma\big(\alpha f + (1-\alpha) g\big)\Big) - 
		    \mathcal{C}^{2}\Big(\sigma\big(\alpha f + (1-\alpha) g\big)\Big).\\
		    & \Longrightarrow \mathcal{C}^{1}\Big(\sigma\big(\alpha f + (1-\alpha) g\big)\Big) \geq  \mathcal{C}^{2}\Big(\sigma\big(\alpha f + (1-\alpha) x\big)\Big)  \\
		    & \ \ \ \ \ \  + \mathcal{C}^{1}\Big(\sigma\big(\alpha f + (1-\alpha) g\big)\Big)-\mathcal{C}^{2}\Big(\sigma\big(\alpha f + (1-\alpha) g\big)\Big) \\
		    & \ \ \ \ \ \  + \mathcal{C}^{1}\Big(\sigma\big(\alpha f + (1-\alpha) g\big)\Big) -	\mathcal{C}^{1}\Big(\sigma\big(\alpha f + (1-\alpha) x\big)\Big) ,
	        \end{split}
            \end{equation*}

            \noindent If $\sigma\big(\alpha f + (1-\alpha) x\big)\subset \sigma\big(\alpha f + (1-\alpha) g\big)$, we have $\mathcal{C}^{1}\Big(\sigma\big(\alpha f + (1-\alpha) g\big)\Big)\geq  	\mathcal{C}^{1}\Big(\sigma\big(\alpha f + (1-\alpha) x\big)\Big) $. Thus, above inequality can be rewrite as
            \begin{equation*}
	        \begin{split}
	        \mathcal{C}^{1}\Big(\sigma\big(\alpha f + (1-\alpha) g\big)\Big) & \geq \mathcal{C}^{2}\Big(\sigma\big(\alpha f + (1-\alpha) x\big)\Big) \\
		    & +\mathcal{C}^{1}\Big(\sigma\big(\alpha f + (1-\alpha) x\big)\Big) - 
	        \mathcal{C}^{2}\Big(\sigma\big(\alpha f + (1-\alpha) g\big)\Big)  \\
		    & + \mathcal{C}^{1}\Big(\sigma\big(\alpha f + (1-\alpha) g\big)\Big) - 
		     \mathcal{C}^{1}\Big(\sigma\big(\alpha f + (1-\alpha) x\big)\Big)\\
		    & =   \mathcal{C}^{2}\Big(\sigma\big(\alpha f + (1-\alpha) x\big)\Big) + 
		    \mathcal{C}^{1}\Big(\sigma\big(\alpha f + (1-\alpha) g\big)\Big) \\
            & - \mathcal{C}^{2}\Big(\sigma\big(\alpha f + (1-\alpha) g\big)\Big),
	        \end{split}
            \end{equation*}

            \noindent By \hyperlink{thm2}{Theorem 2}, we get $\mathcal{C}^{1}\Big(\sigma\big(\alpha f + (1-\alpha) x\big)\Big)\leq  	\mathcal{C}^{2}\Big(\sigma\big(\alpha f + (1-\alpha) g\big)\Big)$. Thus, we have
            \begin{equation*}
	        \begin{aligned}
		    \mathcal{C}^{1}\Big(\sigma\big(\alpha f + (1-\alpha) g\big)\Big)\geq  \mathcal{C}^{2}\Big(\sigma\big(\alpha f + (1-\alpha) x\big)\Big).
	        \end{aligned}
            \end{equation*}

            \indent Therefore, we find this act $h^{\prime}$.\\
            \noindent \textsl{If part}. It is obvious. 
            \end{proof}
            
\section{proofs of Section 4}\label{appB}
\subsection{Proof of Theorem 4.}  \label{appB.1}

            \begin{proof}
            \textsl{Only if} part.  \\
            \indent Step 1. The first part of the proof is the same as the proof of \hyperlink{thm1}{Theorem 1}. 

            \indent Step 2. \textsl{Minimal complexity cost function.} By \hyperlink{mcu}{Axiom 3}, given any $E\in\Sigma^{\prime}$, any $f\in\mathcal{F}$, and $x,z\in\overline{\mathcal{F}} $. If $fEz\succsim_{E} x$, then we have $fEz^{\prime}\succsim_{E} x$ for any $z^{\prime}\in\overline{\mathcal{F}}$ such that $\sigma(fEz)\subset \sigma(fEz^{\prime})$. Then by \hyperlink{lem5}{Lemma 5}, 
            \begin{equation*}
	        \begin{split}
		    & \int_{\Omega} u \big( (fEz)(\omega)\big) \mu_{E}(d\omega)- \mathcal{C}_{E,\mu}\big(\sigma(fEz)\big)\geq \ \int_{\Omega} u ( x(\omega)) \mu_{E}(d\omega)- \mathcal{C}_{E,\mu}\big(\sigma(x)\big)\\
		    & \Longrightarrow \int_{\Omega} u \big( (fEz^{\prime})(\omega)\big) \mu_{E}(d\omega)- \mathcal{C}_{E,\mu}\big(\sigma(fEz^{\prime})\big)\geq \ \int_{\Omega} u ( x(\omega)) \mu_{E}(d\omega)- \mathcal{C}_{E,\mu}\big(\sigma(x)\big)\\
	        \end{split}
            \end{equation*}
            Since $\int_{\Omega} u \big( (fEz)(\omega)\big) \mu_{E}(d\omega)=\int_{\Omega} u \big( f(\omega)\big) \mu_{E}(d\omega)$, we have 
            \begin{equation*}
		    \mathcal{C}_{E,\mu}\big(\sigma(fEz)\big)\geq \mathcal{C}_{E,\mu}\big(\sigma(fEz^{\prime})\big).
            \end{equation*}
            
            However, by monotonicity of cost function, $\sigma(fEz)\subset \sigma(fEz^{\prime})$ implies $\mathcal{C}_{E,\mu}\big(\sigma(fEz)\big)\leq \mathcal{C}_{E,\mu}\big(\sigma(fEz^{\prime})\big)$. Thus, we have $\mathcal{C}_{E,\mu}\big(\sigma(fEz)\big)= \mathcal{C}_{E,\mu}\big(\sigma(fEz^{\prime})\big)$.

            \indent The result holds for any $x,z\in\overline{\mathcal{F}}$ with $\sigma(fEz)\subset \sigma(fEz^{\prime})$. Therefore, we can find another $z^{\prime\prime}$ such that $\sigma(fEz^{\prime\prime})\subset \sigma(fEz)$. If $fEz^{\prime\prime}\succsim_{E} x$, then we have $fEz^{\prime}\succsim_{E} x$. That is $\mathcal{C}_{E,\mu}\big(\sigma(fEz^{\prime\prime})\big)= \mathcal{C}_{E,\mu}\big(\sigma(fEz^{\prime})\big)$. To conclude, there exists $fEz^{*}\in \mathcal{F}^{min}(f,E)$ such that $\mathcal{C}_{E,\mu}\big(\sigma(fEz^{\prime})\big)= \mathcal{C}_{E,\mu}\big(\sigma(fEz^{*})\big)$ for any $z^{\prime}\in\overline{\mathcal{F}}$ with $\sigma(fEz^{*})\subset \sigma(fEz^{\prime})$.  $\hfill \square$
            
            \indent Step 3. \textsl{Ordinal preference consistency.} Follows from the following lemma. 

            \begin{lemma} 
            \hypertarget{lem7}{Consider} a collection of preferences $\big(\succsim, \{\succsim_{E}\}_{E\in \Sigma^{\prime}}\big)$ such that $\succsim$  satisfies \hyperlink{wo}{Axiom 1} and \hyperlink{wm}{Axiom 4} and $\{\succsim_{E}\}_{E\in \Sigma}$ satisfies \hyperlink{dca}{Axiom 9}. Then for all $x,y\in\overline{\mathcal{F}}$ and $E\in\Sigma^{\prime}$, we have
            \begin{equation*} 
            x\succsim y \Longleftrightarrow x\succsim_{E} y.
            \end{equation*} 
            \end{lemma}

            \begin{proof} 
            By completeness, suppose $x\succsim y$. By \hyperlink{wm}{Axiom 4}, $x\succsim y$ implies $xEy\succsim y$ for any $E\in\Sigma$. Then, by \textsl{Dynamic Complexity Aversion}, $xEy\succsim yEy$ implies $x\succsim_{E} y$. Conversely, suppose $x\succsim_{E} y$ and $y\succ x$. Again, by \hyperlink{wm}{Axiom 4}, $y\succ x$ implies $yEx\succsim x$ for any $E\in\Sigma$. With \hyperlink{dca}{Dynamic Complexity Aversion}, we can conclude $y\succ_{E}x$, which contradicts with $x\succsim_{E}y$. Thus, $x\succsim y$. 
            \end{proof}

            \indent Step 4. \textsl{Bayesian Updating.} Suppose for any $E\in\Sigma^{\prime}$, and any $f,g\in\mathcal{F}$ such that $x_{f}=x_{g}$ where $fEx_{f}\in \mathcal{F}^{min}(f,E)$ and $gEx_{g}\in \mathcal{F}^{min}(g,E)$. Then for $x_{g}\in\overline{\mathcal{F}}$ we have
            \begin{equation*}
	        \begin{split}
		    f\succsim_{E} g& \Longleftrightarrow fEx_{g}\succsim gEx_{g}\\
		    &  \Longleftrightarrow \int_{\Omega} u \big( (fEx)(\omega)\big) \mu(d\omega)- \mathcal{C}\big(\sigma(fEx_{g})\big) \\
		    & \ \ \ \ \ \  \ \geq \int_{\Omega} u \big( (gEx)(\omega)\big) \mu(d\omega)- \mathcal{C}\big(\sigma(gEx_{g})\big) \\
		    &  \Longleftrightarrow \int_{E} u\big( f(\omega)\big) \mu(d\omega) +  \int_{\Omega\backslash E} u\big( x(\omega)\big) \mu(d\omega)- \mathcal{C}\big(\sigma(fEx_{g})\big) \\
		    & \ \ \ \ \ \  \ \geq\int_{E} u\big( g(\omega)\big) \mu(d\omega) +  \int_{\Omega\backslash E} u\big( x(\omega)\big) \mu(d\omega)- \mathcal{C}\big(\sigma(gEx_{g})\big) \\
		    &  \Longleftrightarrow \int_{E} u\big( f(\omega)\big) \mu(d\omega) - \mathcal{C}\big(\sigma(fEx_{g})\big) \geq\int_{E} u\big( g(\omega)\big) \mu(d\omega) - \mathcal{C}\big(\sigma(gEx_{g})\big) \\
		    &  \Longleftrightarrow \frac{1}{\mu(E)}\int_{E} u\big( f(\omega)\big) \mu_{E}(d\omega) - \mathcal{C}_{E,\mu}\big(\sigma(f)\big) \geq  \frac{1}{\mu(E)}\int_{E} u\big( g(\omega)\big) \mu_{E}(d\omega) - \mathcal{C}_{E,\mu}\big(\sigma(g)\big) \\
	        \end{split}
            \end{equation*}
            The last equation comes from step 2.

            \noindent \textsl{If} part.  

            \textsl{Minimal Complexity Updating.} With the existence of minimal complexity cost function, we have
            \begin{equation*}
		    \begin{split}
	        \mathcal{C}_{E,\mu}\big(\sigma(fEz)\big) &= \mathcal{C}_{E,\mu}\big(\sigma(fEz^{*})\big).\\
	        & = \mathcal{C}_{E,\mu}\big(\sigma(fEz^{*\prime})\big)\\
		    \end{split}
            \end{equation*}
            for any $z,z^{\prime}\in \overline{\mathcal{F}}$, $\sigma(fEz)\subset \sigma(fEz^{\prime})$, and $fEz^{*}\in \mathcal{F}^{min}(f,E)$. Thus, if $fEz\succsim_{E} fEz^{\prime}$, then 
            \begin{equation*}
	        \begin{split}
            & \int_{\Omega} u \big( (fEz)(\omega)\big) \mu_{E}(d\omega)- \mathcal{C}_{E,\mu}\big(\sigma(fEz)\big)\geq \ \int_{\Omega} u ( x(\omega)) \mu_{E}(d\omega)- \mathcal{C}_{E,\mu}\big(\sigma(x)\big)\\
            & \Rightarrow \int_{\Omega} u \big( f(\omega)\big) \mu_{E}(d\omega)- \mathcal{C}_{E,\mu}\big(\sigma(fEz^{*})\big)\geq \ \int_{\Omega} u ( x(\omega)) \mu_{E}(d\omega)- \mathcal{C}_{E,\mu}\big(\sigma(x)\big)\\
            & \Rightarrow \int_{\Omega} u \big( (fEz^{\prime})(\omega)\big) \mu_{E}(d\omega)- \mathcal{C}_{E,\mu}\big(\sigma(fEz^{*})\big)\geq \ \int_{\Omega} u ( x(\omega)) \mu_{E}(d\omega)- \mathcal{C}_{E,\mu}\big(\sigma(x)\big)\\
            & \Rightarrow \int_{\Omega} u \big( (fEz^{\prime})(\omega)\big) \mu_{E}(d\omega)- \mathcal{C}_{E,\mu}\big(\sigma(fEz^{\prime})\big)\geq \ \int_{\Omega} u ( x(\omega)) \mu_{E}(d\omega)- \mathcal{C}_{E,\mu}\big(\sigma(x)\big)\\
	        \end{split}
            \end{equation*}
            Hence, we have $fEz^{\prime}\succsim_{E}x$. 

            \textsl{Dynamic Complexity Aversion.} For any $E\in \Sigma^{\prime}$, and any $f,g\in \mathcal{F}$, we can find $h\in \mathcal{F}^{min}(g,E)$. Then, there exists $x\in g(E)$ such that $h=gEx$. If $fEx\succsim gEx$, then by \hyperlink{lem5}{Lemma 5} and \hyperlink{mcu}{Axiom 8} we have
            \begin{equation*}
	        \begin{split}
		    & \int_{\Omega} u \big( (fEx)(\omega)\big) \mu(d\omega)- \mathcal{C}\big(\sigma(fEx)\big)\geq \int_{\Omega} u \big( (gEx)(\omega)\big) \mu(d\omega)- \mathcal{C}\big(\sigma(gEx)\big)\\
		    & \Rightarrow  \int_{E} u \big( f(\omega)\big) \mu(d\omega) + \int_{\Omega\backslash E} u \big( x(\omega)\big) \mu(d\omega) - \mathcal{C}\big(\sigma(fEx)\big)\\
		    & \ \ \ \ \geq \int_{E} u \big( g(\omega)\big) \mu(d\omega) + \int_{\Omega\backslash E} u \big( x(\omega)\big) \mu(d\omega) - \mathcal{C}\big(\sigma(gEx)\big)\\
		    & \Rightarrow \int_{E} u \big( f(\omega)\big) \mu(d\omega) - \mathcal{C}\big(\sigma(fEx)\big) \geq \int_{E} u \big( g(\omega)\big) \mu(d\omega) - \mathcal{C}\big(\sigma(gEx)\big)\\
		    & \Rightarrow \int_{E} u \big( f(\omega)\big) \mu_{E}(d\omega) - \frac{1}{\mu(E)}\mathcal{C}\big(\sigma(fEx)\big) \geq \int_{E} u \big( g(\omega)\big) \mu_{E}(d\omega) - \frac{1}{\mu(E)}\mathcal{C}\big(\sigma(gEx)\big)\\
		    & \Rightarrow \int_{E} u \big( f(\omega)\big) \mu_{E}(d\omega) - \int_{E} u \big( g(\omega)\big) \mu_{E}(d\omega)\\
		    & \ \ \ \ \ \geq \frac{1}{\mu(E)}\mathcal{C}\big(\sigma(fEx)\big)- \frac{1}{\mu(E)}\mathcal{C}\big(\sigma(gEx)\big)\\
		    & \ \ \ \ \ \geq\mathcal{C}_{E,\mu}\big(\sigma(f)\big)- \mathcal{C}_{E,\mu}\big(\sigma(g)\big)\\
		    & \Rightarrow \int_{E} u \big( f(\omega)\big) \mu_{E}(d\omega) - \mathcal{C}_{E,\mu}\big(\sigma(f)\big)\geq \int_{E} u \big( g(\omega)\big) \mu_{E}(d\omega) - \frac{1}{\mu(E)}\mathcal{C}_{E,\mu}\big(\sigma(g)\big)\\
	        \end{split}
            \end{equation*}
            The last inequality comes from the two facts: $\mathcal{C}_{E,\mu}\big(\sigma(g)\big)=\frac{1}{\mu(E)}\mathcal{C}\big(\sigma(gEx)\big)$  and $\mathcal{C}_{E,\mu}\big(\sigma(f)\big)=\frac{1}{\mu(E)}\mathcal{C}\big(\sigma(fEx_{f})\big)\geq \frac{1}{\mu(E)}\mathcal{C}\big(\sigma(fEx)\big)$ where $fEx_{f}\in \mathcal{F}^{min}(f,E)$. Therefore, we can conclude $fEx\succsim gEx \Longrightarrow f\succsim_{E} g.$ 

            This completes the proof of necessity of \hyperlink{thm4}{Theorem 4}.
            \end{proof}
            
\subsection{Proofs of Corollary 3} \label{appB.2}
            \begin{proof}
            \noindent \textsl{No act is costless}. By the definition of $\mathcal{C}_{E,\mu}$, given $E\in\Sigma^{\prime}$ and $\mu\in \Delta(\Omega)$, $\mathcal{C}_{E,\mu}(x)=min\{\mathcal{C}(h)/\mu(E): h\in \mathcal{F}, \text{ and }\sigma(h|E)=\sigma(x|E)\}=0$. And $\mathcal{C}_{E,\mu}(f)>0$ for all $f\in\mathcal{F}\backslash \overline{\mathcal{F}}$ and $\sigma(f|E)\neq \{E\}$  

            \indent \textsl{Monotonicity}. Suppose $\sigma(g|E)\subset \sigma(f|E)$, it is easy to see that $\sigma(gEx)\subset \sigma(fEx)$ for any $x\in\overline{\mathcal{F}}$.By the definition of $\mathcal{C}_{E,\mu}$, we have $\sigma(fEx_{f})\subset \sigma(fEx_{g})$ where $fEx_{f}\in \mathcal{F}^{min}(f,E)$. Since $\sigma(gEx_{f})\subset \sigma(fEx_{f})$ and $\sigma(gEx_{g})\subset \sigma(gEx_{f})$ where $gEx_{g}\in \mathcal{F}^{min}(g,E)$, we have $\sigma(gEx_{g})\subset \sigma(fEx_{f})$, which implies $\mathcal{C}\big(\sigma(gEx_{g})\big)\leq \mathcal{C}\big(\sigma(fEx_{f})\big)$. Hence, we have $\mathcal{C}_{E,\mu}\big(\sigma(f)\big)\geq \mathcal{C}_{E,\mu}\big(\sigma(g)\big)$.

            \indent \textsl{Uniqueness}.  Then by the proof of \hyperlink{thm4}{Theorem 4}, $\int_{\Omega} u (f(\omega)\mu_{E}(d\omega)) - \mathcal{C}_{E,\mu}\big(\sigma(f)\big)$ represents $\succsim_{E}$. Suppose for contradiction that there exists another conditional complexity cost function $\mathcal{C}^{\prime}_{E,\mu}$ where $\int_{\Omega} u (f(\omega)\mu_{E}(d\omega)) - \mathcal{C}^{\prime}_{E,\mu}\big(\sigma(f)\big)$ represents $\succsim$. Suppose $\mathcal{C}^{\prime}_{E,\mu}\big(\sigma(f)\big)<\mathcal{C}_{E,\mu}\big(\sigma(f))\big)$ for all $f\in\mathcal{F}^{S^{f}}$ and $\mathcal{C}^{\prime}\big(\sigma(g)\big)=\mathcal{C}\big(\sigma(g)\big)$ for all other acts $g\notin\mathcal{F}^{S^{f}}$.

            \indent For an act $f^{*}\in\mathcal{F}^{S^{f}}$, we construct a new act $f^{\prime}\in\mathcal{F}$ such that $f^{\prime}(\omega) = f^{*}(\omega) + \epsilon$ for all $\omega\in\Omega$ and a very small but positive $u(\epsilon)$. Suppose $f^{\prime}\sim_{E} g$ where $g\in\mathcal{F}$ but $g\notin\mathcal{F}^{S^{f}}$, by $\sigma(f^{*})=\sigma(f^{\prime})$, we know $g\succ_{E} f^{*}$. That is
            \begin{equation*}
	        \begin{split}
		    &\int_{\Omega} u (g(\omega)\mu_{E}(d\omega)) - \mathcal{C}_{E,\mu}\big(\sigma(g)\big)= \int_{\Omega} u (f^{\prime}(\omega)\mu_{E}(d\omega)) - \mathcal{C}_{E,\mu}\big(\sigma(f^{\prime})\big)\\
		    &\Longleftrightarrow \int_{\Omega} u (g(\omega)\mu_{E}(d\omega)) - \mathcal{C}_{E,\mu}\big(\sigma(g)\big) = 
		    \int_{\Omega} u (f^{*}(\omega)\mu_{E}(d\omega)) + u(\epsilon) - \mathcal{C}_{E,\mu}\big(\sigma(f^{*})\big)\\
		    &\Longleftrightarrow \mathcal{C}_{E,\mu}\big(\sigma(f^{*})\big) - u(\epsilon) - \mathcal{C}_{E,\mu}\big(\sigma(g)\big) < \mathcal{C}^{\prime}_{E,\mu}\big(\sigma(f^{*})\big) - \mathcal{C}^{\prime}_{E,\mu}\big(\sigma(g)\big)\\
		    &\Longleftrightarrow  \mathcal{C}_{E,\mu}\big(\sigma(f^{*})\big) - u(\epsilon) < \mathcal{C}^{\prime}_{E,\mu}\big(\sigma(f^{*})\big)\\
	        \end{split}
            \end{equation*}
            This holds for any very small $u(\epsilon)$, so we have $\mathcal{C}_{E,\mu}\big(\sigma(f^{*})\big)=  \mathcal{C}^{\prime}_{E,\mu}\big(\sigma(f^{*})\big)$, a contradiction. 
            \end{proof}

\subsection{Proofs of Corollary 4} \label{appB.3}
            
            \begin{proof} 
            Suppose $\langle u,\mu, \mathcal{C}, \{\mathcal{C}_{E,\mu}\}_{E\in\Sigma^{\prime}}\rangle$ and $\langle u^{\prime}, \mu^{\prime}, \mathcal{C}^{\prime}, \{\mathcal{C}^{\prime}_{E,\mu}\}_{E\in\Sigma^{\prime}}\rangle$ represent the same preferences relations, and $\mathcal{C}$ and $\mathcal{C}^{\prime}$ are defined as in \hyperlink{def1}{Definition 1}.  By \hyperlink{lem2}{Lemma 2}, the preference relation has an expected utility representation, so $\mu=\mu^{\prime}$ and $\exists \beta_{1}>0$ and $\beta_{2}\in \mathbb{R}$ such that $u=\beta_{1} u^{\prime} + \beta_{2}$. 

            \indent Then we turn to prove $\mathcal{C}_{E,\mu} = \alpha \mathcal{C}_{E,\mu}$. Suppose $u$ and $u^{\prime}$ represent the same preference relations. Given any $E\in\Sigma^{\prime}$, consider act $f\in\mathcal{F}$ such that there exists $x\in\overline{\mathcal{F}}$ with $f\succsim_{E} x$. By \hyperlink{lem1}{Lemma 1}, there exist a $x_{f}\in\overline{\mathcal{F}}$ such that $x_{f}\sim_{E} f$. So we have 
            \begin{equation*}
            \int_{\Omega} u (x) \mu_{E}(d\omega) = \int_{\Omega} u \big( f(\omega)\big) \mu_{E}(d\omega)- \mathcal{C}_{E,\mu}\big(\sigma(f)\big)
            \end{equation*}
            and 
            \begin{equation*}
		    \begin{split}
	        & \int_{\Omega} u^{\prime} (x) \mu_{E}(d\omega) = \int_{\Omega} u^{\prime} \big( f(\omega)\big) \mu_{E}(d\omega)- \mathcal{C}^{\prime}_{E,\mu}\big(\sigma(f)\big)\\
	        & \Longrightarrow \alpha\int_{\Omega} u (x) \mu_{E}(d\omega) + \beta = \alpha\int_{\Omega} u \big( f(\omega)\big) \mu_{E}(d\omega) + \beta- \mathcal{C}^{\prime}_{E,\mu}\big(\sigma(f)\big)\\
	        & \Longrightarrow \alpha\big(\int_{\Omega} u (x) \mu_{E}(d\omega) - \int_{\Omega} u \big( f(\omega)\big) \mu_{E}(d\omega) \big) = \mathcal{C}^{\prime}\big(\sigma(f)\big)\\
	        & \Longrightarrow \alpha\mathcal{C}^{\prime}_{E,\mu}\big(\sigma(f)\big) = \mathcal{C}^{\prime}_{E,\mu}\big(\sigma(f)\big). 
		    \end{split}
            \end{equation*}
            \end{proof}
            
\section{Proofs of Section 5} \label{appC}

            \noindent \textbf{Proof of \hyperlink{prop1}{Proposition 1}.} The principal solves the problem with two steps. 

            \begin{proof}
            \noindent Step1: choosing the optimal wage scheme for different complexity level. Consider following three cases.
            
            Case1: $|S^{w}|=1$, that is $W(\omega_{1})=W(\omega_{2})=W(\omega_{3})=W$

            $\sum_{\omega}\mu_{1}(\omega)\mu(\omega)\sqrt{W}=0<c$ violates constraint (2).

            Case2: $|S^{w}|=3$, that is $W(\omega_{1})\neq W(\omega_{2})\neq W(\omega_{3})$

            The Lagrangian for this this problem is 
            \begin{equation*}
	        \begin{split}
	        \mathcal{L} &= -\sum_{\omega}\mu_{1} (\omega)W(\omega)- 3\delta- \lambda \left[\bar{u} - \sum_{\omega}\mu_{1}(\omega)\sqrt{W(\omega)}\right]\\
	        &\ \quad - \beta\left[c-\sum_{\omega}\mu_{1}(\omega)\mu(\omega)\sqrt{W(\omega)} \right]
	        \end{split}
            \end{equation*}
            \indent The first conditions are 
            \begin{equation*}
	     	\frac{\partial \mathcal{L}}{\partial W(\omega)} = -\mu_{1}(\omega) +  \lambda \mu_{1}(\omega)\frac{1}{2\sqrt{W(\omega)}} + \beta \mu_{1}(\omega)\mu(\omega)\frac{1}{2\sqrt{W(\omega)}} = 0, \tag{c1}
            \end{equation*}
            \begin{equation*}
		    \frac{\partial \mathcal{L}}{\partial \lambda} = \bar{u} - \sum_{\omega}\mu_{1}(\omega)\sqrt{W(\omega)} \leq0, \tag{c2}
            \end{equation*}		
            \begin{equation*}
		    \frac{\partial \mathcal{L}}{\partial \beta} = c-\sum_{\omega}\mu_{1}(\omega)\mu(\omega)\sqrt{W(\omega)} \leq0. \tag{c3}
            \end{equation*}
            \indent where (c2) and (c3) hold with equality if $\lambda\neq 0$ and $\beta\neq 0$.

            \indent(c1) can be rewritten as 
            \begin{equation*}
	        \frac{1}{2\sqrt{W(\omega)}}  = \frac{1}{ \lambda +\beta \mu(\omega)}. \tag{c4}
            \end{equation*}
 
            \indent Suppose  that $\beta = 0$. Then (c4) implies $W(\omega)$ is a constant. This means (c3) fails. Thus, $\beta\neq 0$.

            \indent Suppose that $\lambda = 0$.  Then (c4) implies $\sqrt{W(\omega)}$ is negative for at least one $\omega$. Thus, $\lambda \neq 0$.

            \indent Hence, (c2) and (c3) hold with equality. From (c1), we have 
            \begin{equation*}
	        \lambda = 2\bar{u}; \beta =\frac{-2c}{\sum_{\omega}\mu_{0}(\omega)\mu(\omega)}. 
            \end{equation*}
            \indent So we have 
            \begin{equation*}
	        W(\omega) = (\bar{u} - c\mu(\omega)(\sum_{\omega}\mu_{0}(\omega)\mu(\omega))^{-1})^{2},
            \end{equation*}
            and the minimal cost:
            \begin{equation*}
	        \bar{u}^{2} + \frac{c^{2}\sum_{\omega}\mu_{1}(\omega)\mu^{2}(\omega)}{\sum_{\omega}\mu_{0}(\omega)\mu(\omega)} + 3\delta
            \end{equation*}

            Case3: $|S^{w}|=2$, that is $W(\omega_{i})= W(\omega_{j})\neq W(\omega_{k})$ for any $i,j,k\in \{1,2,3\}$

            The Lagrangian for this this problem is 
            \begin{equation*}
	        \begin{split}
		    \mathcal{L} & = -(1-\mu_{1}(\omega_{k}))W - \mu_{1}(\omega_{k})W(\omega_{k})- 2\delta \\ &\ \quad - \lambda \left[\bar{u} - (1-\mu_{1}(\omega_{k}))\sqrt{W} - \mu_{1}(\omega_{k})\sqrt{W(\omega_{k})}\right] \\
		    &\ \quad - \beta\left[c-\mu_{1}(\omega_{k})\mu(\omega_{k})(\sqrt{W(\omega_{k})} - \sqrt{W}) \right]
	        \end{split}
            \end{equation*}

            We have 
            \begin{equation*}
	        (\lambda,\beta) = (2\bar{u}, \beta =\frac{2c}{\mu^{2}(\omega_{k})/(1-\mu_{1}(\omega_{k}))}),
            \end{equation*}
            So we have 
            \begin{equation*}
	        W(\omega_{i}) = W(\omega_{j})= (\bar{u} - c\mu_{1}(\omega_{k})/\mu(\omega_{k}))^{2}. 
            \end{equation*}

            By $\mu_{1}(\omega)\neq \mu_{0}(\omega)$ for at least one $\omega$, there must exist $k\in\{1,2,3\}$ such that $\mu(\omega_{k})<0$. Hence, above result must exist. Moreover, we have
            \begin{equation*}
	        W(\omega_{k}) = (\bar{u} + c(1-\mu_{1}(\omega_{k}))/\mu(\omega_{k}))^{2}. 
            \end{equation*}
            and the minimal cost:
            \begin{equation*}
	        \bar{u}^{2} +\frac{c^{2}\mu_{1}(\omega_{k})(1-\mu_{1}(\omega_{k}))}{\mu^{2}(\omega_{k})} + 2\delta
            \end{equation*}

            \noindent Step2: comparing the design cost among above three cases. 
            The optimal wage scheme is $|S^{W}|=2$ if 
            \begin{equation*}
	        \begin{split}
	        \bar{u}^{2} +\frac{c^{2}\mu_{1}(\omega_{k})(1-\mu_{1}(\omega_{k}))}{\mu^{2}(\omega_{k})} + 2\delta & \leq  \bar{u}^{2} + \frac{c^{2}\sum_{\omega}\mu_{1}(\omega)\mu^{2}(\omega)}{\sum_{\omega}\mu_{0}(\omega)\mu(\omega)} + 3\delta\\
	        & = \frac{c^{2}\sum_{\omega}\mu_{1}(\omega)\mu^{2}(\omega)}{\sum_{\omega}\mu_{0}(\omega)\mu(\omega)} + c^{2} +3\delta - c^{2}\\
	        & = 3\delta - c^{2}
	        \end{split}
            \end{equation*}
            That is
            \begin{equation*}
            \frac{\mu_{1}(\omega_{k})(1-\mu_{1}(\omega_{k}))}{\mu^{2}(\omega_{k})} \leq \delta/c^{2} - 1. 
            \end{equation*}
            \end{proof}

            \noindent \textbf{Proof of \hyperlink{prop2}{Proposition 2}.} 
            \begin{proof}
            It is immediately from the proof of \hyperlink{prop1}{Proposition 1}.
            \end{proof}
            
\section{An Application to Competitive Markets with Uncertainty} \label{appD}

            \indent In this section, I show how the model is embedded in the general equilibrium framework. The model provides an interpretation for extreme prices. Following the general equilibrium literature, I represent uncertianty by assuming that endowments and preferences depend on the state of the world. 

            Formally, consider an economy consisting of $I$ consumers (a typical consumer is denoted $i \in \{1,\cdots, I\}$), and $L$ goods (indexed by $l=1,\cdots,L$). Again, here we take $\Omega$ to be a finite set of states of the world, a typical element is denoted $\omega\in\{\omega_{1},\cdots, \omega_{n}\}$. Let $\mu^{i}_{\omega}$ denote the probability of the state $\omega$ (which could be objective or subjective). Then a state-contingent comsumption vector of consumer $i$ is specified by 
            \begin{equation*}
	        c^{i} = (c^{i}_{1\omega_{1}},\cdots,c^{i}_{L\omega_{1}},\cdots,c^{i}_{1\omega_{n}},\cdots,c^{i}_{L\omega_{n}}).
            \end{equation*}

            It is clearer if we rewrite $c^{i}$ as $c^{i} = (c^{i}_{\omega_{1}},\cdots,c^{i}_{\omega_{n}})$, where $c^{i}_{\omega_{1}}=(c^{i}_{1\omega_{1}},\cdots,c^{i}_{L\omega_{1}})$, i.e., $c^{i}_{\omega}$ is the consumer $i$'s consumption vector under state $\omega$. Similarly, we have consumer $i$'s endowment vector
            \begin{equation*}
	        e^{i} = (e^{i}_{1\omega_{1}},\cdots,e^{i}_{L\omega_{1}},\cdots,e^{i}_{1\omega_{n}},\cdots,e^{i}_{L\omega_{n}}),
            \end{equation*}

            \indent A complexity averse consumer's utility of the consumption vector $c^{i}\in \mathbb{R}^{LS}_{+}$ is 
            \begin{equation*}
	        U(c^{i}) =  \underset{\omega\in \Omega}{\sum}u^{i}_{\omega}(c^{i}_{\omega})\mu^{i}_{\omega} - \mathcal{C}^{i}(c^{i}),
            \end{equation*}
            where $ \mathcal{C}^{i}(\cdot)$ is consumer $i$'s complexity cost function. Here a consumption vector can be view as an act. Preference relations are defined on consumption vectors as what we discussed in section 3. With complexity cost function, a complexity averse consumer $i$ is unable to choose every consumption vector that a standard consumer can choose. Note that this embeds the standard model as a special case, where $\mathcal{C}^{i}(c^{i})=0$ for all $c^{i}\in \mathbb{R}^{LS}_{+} = C$ and $i \in \{1,\cdots, I\}$. This also embeds CCE (\cite{gul2017coarse}) as a special case, where $\mathcal{C}^{i}(c^{i})=0$ for all crude comsumption vectors and $\mathcal{C}^{i}(c^{i})=\infty$ for all others. 

            \begin{definition} 
            \textsl{Given an economy} $\mathcal{E}$ \textsl{specified by} $\{u^{i},e^{i},\mathcal{C}^{i},\mu^{i}\}^{I}_{i=1}$, \textsl{an allocation} $(\hat{c}^{1}, \cdots, \hat{c}^{I})$, \textsl{and a price vector} $p=(p_{1\omega_{1}},\cdots,p_{L\omega_{n}})$ \textsl{constitute a} \textbf{complexity averse competitive equilibrium}(CACE)  \textsl{if} 
            \begin{enumerate}
	        \renewcommand{\labelenumi}{(\theenumi)}
	        \item \textsl{For every} $i$, $\hat{c}^{i}\in C^{i}$ \textsl{is maximal for} $\succsim^{i}$ \textsl{in the budget set}
	        \begin{equation*}
		    \{c^{i}\in C^{i}: p\cdot c^{i}\leq p\cdot e^{i}\};
	        \end{equation*}
	        \item $\sum_{i}\hat{c}^{i} = \sum_{i}e^{i}$.\\
            \end{enumerate}
            \end{definition}

            \indent To illustrate the idea, consider the following example of consumption under uncertainty. Consider an economy $\mathcal{E}$ with  $I= 3,\ L=1$, and  $\Omega=\{\omega_{1},\omega_{2},\omega_{3}\}$. Suppose we have $e^{1}=(1,0,0)$, $e^{2}=(0,2,0)$, $e^{3}=(0,0,3)$, and utility index of the form $u^{i}=\ln c^{i}_{\omega}$. Thus, an agent's utility of her consumption vector $c^{i}$ is $U(c^{i})=\mu^{i}_{\omega_{1}}\ln c^{i}_{\omega_{1}}+\mu^{i}_{\omega_{2}}\ln c^{i}_{\omega_{2}}+\mu^{i}_{\omega_{3}}\ln c^{i}_{\omega_{3}}$, where $(\mu^{i}_{\omega_{1}},\mu^{i}_{\omega_{2}},\mu^{i}_{\omega{3}})=(\frac{1}{3},\frac{1}{3},\frac{1}{3})$ are the agent $i$'s subjective probabilities of the three states. 

            \indent Since there are three states, consumers could have five different information structures 
            \begin{equation*}
	        (\{1,2,3\}), (\{1\},\{2\},\{3\}), (\{1,2\},\{3\}), (\{1\},\{2,3\}), (\{1,3\},\{2\}). 
            \end{equation*}
            Thus, every agent has three types of consumption vectors $C\in \{C_{1},C_{2},C_{3}\}$. $C_{1}$ means that the agent chooses to make a most complex consumption vector where $c^{i}_{\omega}\neq c^{i}_{\omega^{\prime}}$ for any $\omega,\omega^{\prime}\in \Omega$ corresponding to information structure $(\{\omega_{1}\},\{\omega_{2}\},\{\omega_{3}\})$; $C_{2}$ is the coarsest consumption vector where $c^{i}_{\omega}= c^{i}_{\omega^{\prime}}$ for any $\omega,\omega^{\prime}\in \Omega$ corresponding to information structure $(\{\omega_{1},\omega_{2},\omega_{3}\})$; $C_{3}$ indicates that the agent is able to make some complex consumption vectors where $c^{i}_{\omega}=c^{i}_{\omega^{\prime}}\neq c^{i}_{\omega^{\prime\prime}}$ for any $\omega,\omega^{\prime},\omega^{\prime\prime}\in \Omega$. Preferences are well defined, so the agent is capable of identifying the cost for all consumption vectors. By assumption, we have $\mathcal{C}(\hat{c}^{ij})=\mathcal{C}(\hat{c}^{ik})\ \forall \hat{c}^{ij},\hat{c}^{ik}\in C_{1}$, $\mathcal{C}(\bar{c}^{ij})=\mathcal{C}(\bar{c}^{ik})\ \forall \bar{c}^{ij},\bar{c}^{ik}\in C_{3}$, and $\mathcal{C}(\tilde{c}^{ij})=\mathcal{C}(\tilde{c}^{ik})=0\ \forall \hat{c}^{ij},\tilde{c}^{ik}\in C_{2}$. Moreover, $\mathcal{C}(\hat{c}^{i})>\mathcal{C}(\bar{c}^{i})>  \mathcal{C}(\tilde{c}^{i})$ for any $\hat{c}^{i}\in C_{1 },\bar{c}^{i}\in C_{3},\tilde{c}^{i}\in C_{2}$. 

            \indent If all consumers choose consumption vectors in $C_{1}$, this is the standard competitive equilibrium (SCE), we have $c^{*} = (c^{i}_{1},c^{i}_{2}, c^{i}_{3})=(\frac{1}{3},\frac{2}{3},1)$ for all $i\in I$, $\delta=(1,0,0)$, and the prices are given by $p^{*}=(p_{1}, p_{2}, p_{3})=(1,\frac{1}{2},\frac{1}{3})$. Note that when consumer 1 chooses a plan in $C_{1}$, her utility is $\frac{1}{3}\ln\big(\frac{(p_{1})^{2}}{27p_{2}p_{3}}\big)-\mathcal{C}(C_{1})$. If consumer 1 chooses a consumption vector where $c^{1}_{1}= c^{1}_{2}\neq c^{1}_{3}$, it is easy to verify that her utility is  $\frac{1}{3}\ln\big(\frac{4(p_{1})^{3}}{27(p_{1}+p_{2})^{2}p_{3}}\big)-\mathcal{C}(C_{3})$. If we let $\mathcal{C}(C_{1})-\mathcal{C}(C_{3}) > \frac{1}{3}\ln(\frac{(p_{1}+p_{2})^{2}}{4p_{1}p_{2}})$, then $(c^{*},p^{*})$ is not an equilibrium any more. Suppose consumer 1 is complexity averse and her best choice is to pick consumption vectors in $C_{3}$ with specific cost function $\mathcal{C}^{1}$, and consumer 2 and 3 are standard agents. Then, we can find a CACE $(\hat{c}, \hat{p})=\Big(\big((\frac{1}{2},\frac{1}{2},1.2),(\frac{2}{9},\frac{2}{3},\frac{4}{5}),(\frac{5}{18},\frac{5}{6},1)\big);\big(1,\frac{1}{3},\frac{5}{18}\big)\Big)$. There is a significant difference between SCE price $p^{*}$ and CACE price $\hat{p}$. We can see that $\frac{p^{*}_{1}}{p^{*}_{3}}<\frac{\hat{p}_{1}}{\hat{p}_{3}}$, which indicates that CACE prices might be more extreme than SCE prices.
\end{appendices}

\end{document}